\def\llncs{0}
\def\fullpage{1}
\def\anonymous{0}
\def\authnote{0}
\def\notxfont{0}
\def\submission{0}

\def\inclappndx{1}
\def\compressed{0}

\ifnum\submission=1
\def\llncs{1}
\else
\def\inclappndx{1}
\fi

\ifnum\llncs=1
\documentclass[envcountsect,runningheads]{llncs}
\else
	\documentclass[letterpaper,hmargin=1.05in,vmargin=1.05in,11pt]{article}
			\ifnum\fullpage=1
		\usepackage{fullpage}
		\fi
\fi

\usepackage[%
  colorlinks=true,
  citecolor=blue,
  pagebackref=true
]{hyperref}

\usepackage{amsmath, amsfonts, amssymb, mathtools,amscd}

\usepackage{amsthm}

\usepackage{lmodern}
\usepackage[T1]{fontenc}
\usepackage[utf8]{inputenc}

\usepackage{etex}

\usepackage{arydshln} 
\usepackage{url}
\usepackage{ifthen}
\usepackage{bm}
\usepackage{multirow}
\usepackage[dvips]{graphicx}
\usepackage[usenames]{color}
\usepackage{xcolor,colortbl} 
\usepackage{threeparttable}
\usepackage{comment}
\usepackage{paralist,verbatim}
\usepackage{cases}
\usepackage{booktabs}
\usepackage{braket}
\usepackage{cancel} 
\usepackage{ascmac} 
\usepackage{framed}
\ifnum\compressed=0
    \usepackage{authblk}
\fi
\usepackage{pifont}
\usepackage{qcircuit}
\definecolor{darkblue}{rgb}{0,0,0.6}
\definecolor{darkgreen}{rgb}{0,0.5,0}
\definecolor{maroon}{rgb}{0.5,0.1,0.1}
\definecolor{dpurple}{rgb}{0.2,0,0.65}

\usepackage[capitalise,noabbrev]{cleveref}
\usepackage[absolute]{textpos}
\usepackage[final]{microtype}
\usepackage[absolute]{textpos}
\usepackage{everypage}
\DeclareMathAlphabet{\mathpzc}{OT1}{pzc}{m}{it}

\usepackage{algorithmic}
\usepackage{algorithm}
\usepackage{here}

\newtheoremstyle{thicktheorem}%
{\topsep}
{\topsep}
{\itshape}{}%
{\bfseries}%
{.}
{ }%
{\thmname{#1}\thmnumber{ #2}%
		\thmnote{ (#3)}%
}

\newtheoremstyle{remark}
{\topsep}
{\topsep}
	{}
	{}
	{}
	{.}
	{ }
	{\textit{\thmname{#1}}\thmnumber{ #2}
			\thmnote{ (#3)}%
	}

\ifnum\llncs=0
	\theoremstyle{thicktheorem}
	\newtheorem{theorem}{Theorem}[section]
	\newtheorem{lemma}[theorem]{Lemma}
	\newtheorem{corollary}[theorem]{Corollary}
	
	\newtheorem{definition}[theorem]{Definition}

	\theoremstyle{remark}
	
	\newtheorem{remark}[theorem]{Remark}

\else
\fi

\Crefname{MyClaim}{Claim}{Claims}

	\crefname{theorem}{Theorem}{Theorems}
	\crefname{assumption}{Assumption}{Assumptions}
	\crefname{construction}{Construction}{Constructions}
	\crefname{corollary}{Corollary}{Corollaries}
	\crefname{conjecture}{Conjecture}{Conjectures}
	\crefname{definition}{Definition}{Definitions}
	\crefname{exmaple}{Example}{Examples}
	\crefname{experiment}{Experiment}{Experiments}
	\crefname{counterexample}{Counterexample}{Counterexamples}
	\crefname{lemma}{Lemma}{Lemmata}
	\crefname{observation}{Observation}{Observations}
	\crefname{proposition}{Proposition}{Propositions}
	\crefname{remark}{Remark}{Remarks}
	\crefname{claim}{Claim}{Claims}
	\crefname{fact}{Fact}{Facts}
	\crefname{note}{Note}{Notes}

\ifnum\llncs=1
 \crefname{appendix}{App.}{Appendices}
 \crefname{section}{Sec.}{Sections}
\else
\fi

\ifnum\llncs=1
\renewcommand*{\backref}[1]{}
\else
	\renewcommand*{\backref}[1]{(Cited on page~#1.)}
	\ifnum\notxfont=1
	\else
		\usepackage{newtxtext}
	\fi
\fi
\ifnum\authnote=0  
\newcommand{\mor}[1]{}
\newcommand{\barak}[1]{}
\newcommand{\takashi}[1]{} 

\else
\newcommand{\mor}[1]{$\ll$\textsf{\color{red} Tomoyuki: { #1}}$\gg$}
\newcommand{\takashi}[1]{$\ll$\textsf{\color{orange} Takashi: { #1}}$\gg$}
\newcommand{\barak}[1]{$\ll$\textsf{\color{purple} Barak: { #1}}$\gg$}

\fi

\newcommand*{\qreg}[1]{{\color{gray}{\mathsf{#1}}}}

\newcommand{\calY}{\mathcal{Y}}
\newcommand{\calX}{\mathcal{X}}

\newcommand{\Tr}{\mathrm{Tr}}

\newcommand{\NP}{\mathbf{NP}}

\newcommand{\la}{\leftarrow}
\newcommand{\ra}{\rightarrow}

\newcommand{\concat}{\|}

\newcommand{\cA}{\mathcal{A}}
\newcommand{\cB}{\mathcal{B}}

\newcommand{\cH}{\mathcal{H}}

\newcommand{\cK}{\mathcal{K}}

\newcommand{\cV}{\mathcal{V}}

\def\makeuppercase#1{
\expandafter\newcommand\csname tl#1\endcsname{\widetilde{#1}}
}

\def\makelowercase#1{
\expandafter\newcommand\csname tl#1\endcsname{\widetilde{#1}}
}

\newcommand{\regC}{\mathbf{C}}

\newcommand{\regR}{\mathbf{R}}

\newcommand{\secp}{\lambda}

\newcommand{\A}{\entity{A}}
\newcommand{\B}{\entity{B}}

\newcommand*{\algo}[1]{\ensuremath{\mathsf{#1}}}

\newcommand*{\entity}[1]{\mathcal{#1}}

\newenvironment{boxfig}[2]{\begin{figure}[#1]\fbox{\begin{minipage}{0.97\linewidth}
                        \vspace{0.2em}
                        \makebox[0.025\linewidth]{}
                        \begin{minipage}{0.95\linewidth}
            {{
                        #2 }}
                        \end{minipage}
                        \vspace{0.2em}
                        \end{minipage}}}{\end{figure}}

\newcommand{\bit}{\{0,1\}}

\newcommand{\fail}{\mathtt{fail}}

\newcommand{\setup}{\algo{Setup}}

\newcommand{\TD}{\algo{TD}}

\newcommand{\negl}{{\mathsf{negl}}}

\newcommand{\poly}{{\mathrm{poly}}}

\newcommand{\lang}{\mathcal{L}}
\newcommand{\rel}{\mathcal{R}}
\newcommand{\pro}{P}
\newcommand{\ver}{V}
\newcommand{\execution}[2]{\langle #1, #2 \rangle}
\newcommand{\OUT}{\mathsf{OUT}}

\makeatletter
\DeclareRobustCommand
  \myvdots{\vbox{\baselineskip4\p@ \lineskiplimit\z@
    \hbox{.}\hbox{.}\hbox{.}}}
\makeatother

\title{
Unconditionally Secure Commitments with Quantum Auxiliary Inputs 
}

\ifnum\anonymous=1
\ifnum\llncs=1
\author{\empty}\institute{\empty}
\else
\author{}
\fi
\else
\ifnum\llncs=1
\author{
Tomoyuki Morimae\inst{1} 
\and Barak Nehoran\inst{2}
\orcidID{0000-0001-7371-0829} 
\and Takashi Yamakawa\inst{3,4,1} 
\orcidID{0000-0003-1712-3026}
}
\institute{
    Yukawa Institute for Theoretical Physics, Kyoto University, Kyoto, Japan 
    \email{tomoyuki.morimae@yukawa.kyoto-u.ac.jp}\\
    \and 
    Princeton University, Princeton, NJ, USA\\ 
    \email{bnehoran@princeton.edu}\\
    \url{https://www.cs.princeton.edu/~bnehoran/}
    \and 
    NTT Social Informatics Laboratories, Tokyo, Japan
    \and 
    NTT Research Center for Theoretical Quantum Information, Atsugi, Japan
    \email{takashi.yamakawa@ntt.com}\\
    \url{https://sites.google.com/view/takashiyamakawa}
}
\else
\author[1]{Tomoyuki Morimae}
\author[2]{Barak Nehoran}
\author[3,4,1]{Takashi Yamakawa}
\affil[1]{{\small Yukawa Institute for Theoretical Physics, Kyoto University, Kyoto, Japan}\authorcr{\small tomoyuki.morimae@yukawa.kyoto-u.ac.jp}}
\affil[2]{{\small 
Princeton University, Princeton, NJ, USA
}\authorcr{\small bnehoran@princeton.edu}}
\affil[3]{{\small NTT Social Informatics Laboratories, Tokyo, Japan}\authorcr{\small 
takashi.yamakawa@ntt.com}}
\affil[4]{{\small NTT Research Center for Theoretical Quantum Information, Atsugi, Japan}}

\fi 
\fi

\date{}

\begin{document}

\maketitle

\begin{abstract}
We show the following unconditional results on quantum commitments in two related yet different models:  
\begin{enumerate}
\item We revisit the notion of quantum auxiliary-input commitments introduced by Chailloux, Kerenidis, and Rosgen (Comput. Complex. 2016) where  both the committer and receiver take the same quantum state, which is determined by the security parameter, 
as quantum auxiliary inputs. 
We show that computationally-hiding and statistically-binding quantum auxiliary-input commitments exist unconditionally, i.e., without relying on any unproven assumption, while Chailloux et al. assumed a complexity-theoretic assumption,
${\bf QIP}\not\subseteq{\bf QMA}$. 
On the other hand, we observe that achieving both statistical hiding and statistical binding at the same time is impossible even in the quantum auxiliary-input setting. 
To the best of our knowledge, this is the first example of unconditionally proving computational security of 
any form of (classical or quantum) commitments 
for which statistical security is impossible. 
As intermediate steps toward our construction, we introduce and unconditionally construct post-quantum sparse pseudorandom distributions and quantum auxiliary-input EFI pairs which may be of independent interest.
\item  
We introduce a new model which we call the common reference quantum state (CRQS) model where both the committer and receiver take the same quantum state that is randomly sampled by an efficient setup algorithm. 
We unconditionally prove that there exist statistically hiding and statistically binding commitments in the CRQS model, circumventing the impossibility in the plain model.  
\end{enumerate}
We also discuss their applications to zero-knowledge proofs, oblivious transfers, and multi-party computations.
\end{abstract}

\section{Introduction}
Commitments are one of the most fundamental primitives in cryptography. 
A committer commits a bit string to a receiver, and later it is opened.
The receiver cannot learn the committed message until it is opened (hiding),
and the committer cannot change the message once it is committed (binding).
It is easy to see that in classical cryptography, it is impossible to achieve both statistical hiding and statistical binding 
at the same time (i.e., hiding and binding hold against computationally unbounded adversaries).\footnote{Let $\mathsf{Comm}(b,x)$ be the commitment of the bit $b$ with the decommitment $x$.
Because it is statistically hiding, there should be another $x'$ such that $\mathsf{Comm}(b\oplus 1,x')=\mathsf{Comm}(b,x)$.
However, then, an unbounded malicious committer can compute $x'$ to break the binding.}
It is likewise well-known that the no-go holds even in quantum cryptography where quantum computing and quantum communication are possible~\cite{LoChau97,Mayers97}.
Thus, there has been an extensive body of research on quantum commitments with computational hiding or binding where the computational power of the adversary is assumed to be quantum polynomial-time (QPT) ~\cite{EC:DumMaySal00,EC:CreLegSal01,KO09,KO11,YWLQ15,C:MorYam22,C:AnaQiaYue22,EC:HhaMorYam23}. 

This raises the question of what is the minimal complexity assumption that would imply quantum commitments. 
In the classical setting, the existence of commitments is known to be equivalent to the existence of one-way functions~\cite{SIAMCOMP:HILL99,Nao91}. 
In the quantum setting, Brakerski, Canetti, and Qian~\cite{ITCS:BraCanQia23} (based on an earlier work by Yan~\cite{AC:Yan22}) recently showed that the existence of quantum commitments is equivalent to the existence of \emph{EFI pairs}, which are pairs of two efficiently generatable quantum states that are statistically far but computationally indistinguishable. 
Some useful quantum cryptographic primitives such as oblivious transfers (OTs) and multiparty computations (MPCs)
can be constructed from EFI pairs~\cite{C:BCKM21b,EC:GLSV21,C:MorYam22,C:AnaQiaYue22}. Moreover, various quantum cryptographic primitives imply EFI pairs including pseudorandom states generators~\cite{C:JiLiuSon18}, private-key quantum money with pure banknotes~\cite{C:JiLiuSon18,cryptoeprint:2023/1620},
one-way states generators with pure output states~\cite{C:MorYam22,cryptoeprint:2023/1620},  
and one-time-secure secret-key encryption~\cite{cryptoeprint:2022/1336}.
The ultimate goal would be to prove the existence of EFI pairs \emph{unconditionally}, i.e., without relying on any unproven assumption, which would lead to the unconditional existence of commitments (and, consequently, 
OTs and MPCs).\footnote{One might mistakenly assume that because statistically-secure commitments are impossible, commitments and EFI pairs may not exist unconditionally. However, this is not correct: The unconditional existence (i.e., existence without any assumption) of a cryptographic primitive, and its statistical security (i.e., security against unbounded adversaries), are not the same and should not be confused. Here, we mean the unconditional existence (without relying on an unproven complexity assumption) of commitments that are \emph{computationally} secure (i.e., either hiding or binding is secure only against QPT adversaries).}

There has been recent progress in understanding the computational complexity of EFI pairs. For instance, their complexity cannot be captured within traditional complexity classes of classical inputs and outputs~\cite{cryptoeprint:2023/1602}, and instead requires the study of unitary complexity classes~\cite{BEMPQY23}.
It is yet unclear, however, if known barriers to proving complexity separations from the classical theory may reappear in similar form in the quantum case (though the barriers are not directly implied), or if the distinctive nature of unitary complexity classes allows such results to be within reach of present techniques. In this work, we demonstrate that the complexity assumptions previously thought necessary can in fact be removed from commitments of a special non-standard form: those that take quantum auxiliary inputs.

\paragraph{\bf (Quantum) auxiliary-input commitments.}
Chailloux, Kerenidis, and Rosgen~\cite{CKR16} studied computationally-secure quantum commitments in an \emph{auxiliary-input} setting where the honest 
committer and receiver can take classical or quantum auxiliary inputs (which are not necessarily efficiently generatable) for executing the protocol.\footnote{Quantum auxiliary inputs are called \emph{quantum advice} in \cite{CKR16}. We use ``auxiliary input'' and ``advice'' interchangeably.} Such auxiliary-input versions of cryptographic primitives have often been studied when we only assume worst-case assumptions that do not imply conventional notions of cryptographic primitives~\cite{OW93,SIAM:Vadhan06,ITCS:Nanashima21}.  
The paper \cite{CKR16} obtained feasibility results on classical and quantum auxiliary-input quantum commitments based on some worst-case complexity assumptions.  
Specifically, they showed that: 
\begin{itemize}
\item if $\mathbf{QSZK}\nsubseteq \mathbf{QMA}$, then there are classical auxiliary-input quantum commitments (that are statistically binding and computationally hiding), and
\item if $\mathbf{QIP}\nsubseteq \mathbf{QMA}$, then there are quantum auxiliary-input commitments\footnote{When we consider the quantum auxiliary-input setting, we omit ``quantum'' before ``commitments'' since it is clear that they must be quantum commitments given that the auxiliary input is quantum.}
 (both statistically
hiding and computationally binding as well as statistically binding
and computationally hiding).
\end{itemize}
Brakerski, Canetti, and Qian~\cite{ITCS:BraCanQia23} gave an alternative construction of classical auxiliary-input quantum commitments assuming $\mathbf{QCZK}\nsubseteq \mathbf{BQP}$.  
This is, however, not a viable path to proving the unconditional existence of quantum auxiliary-input commitments, since proving any of these complexity-theoretic separations \emph{unconditionally} seems still out of reach of the current knowledge of complexity theory.
In this work, we show, perhaps surprisingly, that such complexity-theoretic assumptions can be removed completely.

\subsection{Our Results}
\paragraph{\bf Quantum auxiliary-input commitments.}
We show that quantum auxiliary-input commitments exist \emph{unconditionally}, i,e., without relying on any unproven assumption. 
\begin{theorem}
There exist quantum auxiliary-input commitments that satisfy computational hiding and statistical binding.
\end{theorem}
We stress that the unconditional construction does not mean that our construction satisfies statistical hiding and statistical binding. 
Rather, we prove that our construction satisfies \emph{computational} hiding and statistical binding without relying on any unproven assumption.   
Indeed, we show that it is impossible to achieve both statistical hiding and statistical binding simultaneously even in the quantum auxiliary-input setting by extending the impossibility result in the plain model~\cite{LoChau97,Mayers97}.
To the best of our knowledge, this is the first example of 
any form of (classical or quantum) commitments 
for which statistical security is impossible but computational security is proven unconditionally.

\paragraph{\bf Post-quantum sparse pseudorandom distributions and quantum auxiliary-input EFI.}
As intermediate steps to show the above theorem, we introduce two new primitives, namely,
a post-quantum version of sparse pseudorandom distributions defined in~\cite{DBLP:journals/rsa/GoldreichK92} 
and a quantum auxiliary-input version of EFI pairs,\footnote{Quantum auxiliary-input EFI pairs are mentioned in \cite{ITCS:BraCanQia23} without a formal definition.} and show their existence unconditionally. 
\begin{theorem}
There exist post-quantum sparse pseudorandom distributions and
quantum auxiliary-input EFI pairs. 
\end{theorem}
A sparse pseudorandom distribution is a (not necessarily efficiently samplable) classical distribution supported by a sparse subset but indistinguishable from the uniform distribution against
classical non-uniform polynomial-time distinguishers.
Goldreich and Krawczyk 
\cite{DBLP:journals/rsa/GoldreichK92} showed the existence of such distributions unconditionally.
We extend sparse pseudorandom distributions to the post-quantum setting where the distinguishers are QPT with quantum advice,
and construct them unconditionally.
Quantum auxiliary-input EFI pairs are the same as the EFI pairs except that the two states
are not necessarily efficiently generatable. We construct quantum auxiliary-input EFI pairs from post-quantum sparse pseudorandom distributions,
which shows the unconditional existence of quantum auxiliary-input EFI pairs.
Interestingly, in the classical case, it is not known 
how to construct (even auxiliary-input) classical commitments from classical sparse pseudorandom distributions.

\paragraph{\bf Application to zero-knowledge.}
As an application of our quantum auxiliary-input commitments, we plug them into Blum's Hamiltonicity protocol~\cite{Blu87} to obtain the following theorem.
\begin{theorem}
There exist zero-knowledge proofs for $\mathbf{NP}$ in the quantum auxiliary-input setting with non-uniform simulation (with quantum advice) and soundness error $1/2$.
\end{theorem}
We can also use our quantum auxiliary-input commitments to instantiate the 3-coloring protocol of \cite{SICOMP:GolMicRac89} and the quantum $\Sigma$-protocol for $\mathbf{QMA}$ of \cite{SIAM:BG22}.  
On the other hand, unfortunately, we do not know how to instantiate the construction of OTs of ~\cite{C:BCKM21b} using our quantum auxiliary-input commitments.  
This is due to the fact that there may not be an efficient way to generate the quantum auxiliary-input, which prevents us from applying Watrous' rewinding lemma~\cite{SIAM:Watrous09} that is used in the security proof in ~\cite{C:BCKM21b}.\footnote{A knowledgeable reader may wonder why we could prove quantum zero-knowledge even though the original proof in the plain model also uses Watrous' rewinding lemma. 
This is because we can avoid using Watrous' rewinding lemma in the fully non-uniform simulation setting (see \Cref{sec:overview}).
}

\paragraph{\bf Commitments in the CRQS model.} 
Our result on quantum auxiliary-input commitments is theoretically interesting. However, the fact that there is no efficient way to generate the quantum auxiliary input makes it unlikely to have uses in real-world applications.
We therefore consider an alternative model that involves only efficiently generatable states. 
Specifically, we introduce a new notion that we call the {\it common reference quantum state (CRQS) model},\footnote{\cite{cryptoeprint:2022/435} also introduced a model which they call the CRQS model, but their model is different from ours. 
In their CRQS model, parties may take arbitrarily entangled quantum states as setup. In our opinion, this is a quantum analog of the correlated randomness model~\cite{TCC:IKMOP13} rather than the common reference string model.
} 
where an efficient setup algorithm randomly samples a classical key $k$ and then
distributes many copies of a (pure)\mor{why you restrict to pure?}\takashi{I feel that's more natural analog of classical CRS.
If we allow mixed, then that seems to be a generalization of classical common reference \emph{distribution} model.
}\mor{But pure state also represents some distribution. Anyway, feel free to go ahead. Or may be you add some footnote?} quantum state $\ket{\psi_k}$ associated with the key $k$. 
It is a natural quantum analog of the common reference string model in classical cryptography.

At first glance, the CRQS model may look similar to the quantum auxiliary-input setting since in both settings, the committer and receiver receive some quantum state as a resource for executing the protocol. However, the crucial differences are that:
\begin{enumerate}
\item quantum auxiliary input may not be efficiently generatable, but a CRQS must be efficiently generated by the setup algorithm, and  
\item whereas quantum auxiliary input must be a fixed quantum state that is determined by the security parameter, the CRQS is instead associated with a randomly sampled classical key $k$ that is hidden from the adversary.\footnote{Later we will see that this property enables the CRQS model to realize statistically-secure commitments (i.e., both hiding and binding are statistical), a feat which is not possible in the auxiliary-input setting.}
\end{enumerate}
Thus, the two settings are incomparable. Nonetheless, we can use a similar idea to construct commitments in the CRQS model unconditionally:
\begin{theorem}
There exist commitments in the CRQS model that satisfy statistical hiding against adversaries that are given bounded polynomial number of copies of the CRQS $\ket{\psi_k}$ and statistical binding even against adversaries that are given the classical key $k$.\footnote{The honest receiver requires only a single copy of the CRQS $\ket{\psi_k}$, but we consider hiding security even against a malicious receiver that collects many copies of $\ket{\psi_k}$. Similarly, the honest committer requires only a single copy of the CRQS $\ket{\psi_k}$, but we consider binding security even against a malicious committer that knows the secret key $k$. These relaxations just make the theorem stronger.}   
\end{theorem} 
We remark that our scheme in the CRQS model satisfies both statistical hiding and statistical binding simultaneously unlike the case in the quantum auxiliary-input setting, where the hiding is only computational. This is made possible by the assumption that the number of copies of the CRQS given to the hiding adversary (malicious receiver) is bounded. We explain why the no-go result of~\cite{LoChau97,Mayers97} does not extend to this setting in \Cref{sec:overview}. 
It is worth noting that the no-go result \emph{does} extend to quantum commitments in the (classical) common reference string (CRS) model\footnote{In the CRS model, a bit string $x$ is sampled by the setup algorithm, and the same $x$ is sent to both the committer and receiver.}, which demonstrates a fundamental difference between the CRQS and CRS models. More generally, we show that the no-go result extends to the correlated randomness model,\footnote{In the correlated randomness model, a pair of bit strings $(x,y)$ is sampled by the setup algorithm, $x$ is sent to the committer, and $y$ is sent to the receiver.} if the correlation of the strings sent to the committer and the receiver is too strong or too weak. 
See 
\ifnum\inclappndx=0
the full version of this paper~\cite{cryptoeprint:2023/1844}
\else
\Cref{sec:nogo_correlatedRS} 
\fi
for details.

\paragraph{\bf Applications to OTs and MPCs.}
As an application of our quantum commitments in the CRQS model, we can plug them into the compiler of \cite{C:BCKM21b} to obtain the following theorem: 
\begin{theorem} 
There exist maliciously simulation-secure OTs in the CRQS model that is statistically secure against adversaries that are given bounded number of copies of CRQS.
\end{theorem}
We note that we can apply the compiler of \cite{C:BCKM21b} unlike in the quantum auxiliary-input setting since the CRQS $\ket{\psi_k}$ can be efficiently generated given the classical key $k$, which can be sampled by the simulator itself.\footnote{Note that we allow the simulator to ``program'' the CRQS similarly to conventional simulation-based security in the classical CRS model. See 
\ifnum\inclappndx=0
the full version of this paper~\cite{cryptoeprint:2023/1844}
\else
\Cref{def:OT_CRQS} 
\fi
for the precise definition of the simulation-based security for OTs in the CRQS model.}   
Moreover it is known that OTs imply MPCs for classical functionalities~\cite{C:IshPraSah08} or even for quantum functionalities~\cite{EC:DGJMS20} in a black-box manner. Thus, we believe that similar constructions work in the CRQS model, which would lead to  statistically maliciously simulation-secure MPCs in the CRQS model. However, for formally showing them, we have to carefully reexamine these constructions to make sure that they work in the CRQS model as well. This is out of the scope of this work, and we leave it to future work. \takashi{I'm not super familiar with MPC and cannot state these results confidently.}

\subsection{Technical Overview}\label{sec:overview}
\paragraph{\bf Post-quantum sparse pseudorandom distributions.} 
We start from recalling the notion of sparse pseudorandom distributions introduced by Goldreich and Krawczyk~\cite{DBLP:journals/rsa/GoldreichK92}. We say that a family $\{D_\secp\}_{\secp\in \mathbb{N}}$ of distributions $D_\secp$ over $\bit^\secp$ is a sparse pseudorandom distribution if it satisfies the following:
\begin{itemize}
\item[{\bf Sparseness:}] The support of $D_\secp$ is sparse, i.e., $\frac{|\mathsf{Supp}(D_\secp)|}{2^\secp}=\negl(\secp)$, where $\mathsf{Supp}(D_\secp)$ denotes the support of $D_\secp$;
\item[{\bf Pseudorandomness:}] A string sampled from $D_\secp$ is computationally indistinguishable from a uniformly random $\secp$-bit string against non-uniform \emph{classical} polynomial-time distinguishers.
\end{itemize}
We stress that $D_\secp$ is not assumed to be efficiently samplable. 
Goldreich and Krawczyk~\cite{DBLP:journals/rsa/GoldreichK92} gave an unconditional proof for the existence of sparse pseudorandom distributions. 
Their proof is based on a simple counting argument where they rely on the fact that the number of possible classical circuits of size $s$ is at most $2^{s^2}$. 

We consider the post-quantum version of sparse pseudorandom distributions where we require pseudorandomness against non-uniform QPT distinguishers. 
If we only consider QPT distinguishers with \emph{classical} advice, then a similar counting argument works. However, when we consider QPT distinguishers with \emph{quantum} advice, then such a simple counting argument no longer works since the number of possible $s$-qubit states is \emph{double exponential} rather than exponential in $s$ even if we consider the $\epsilon$-net with sufficiently small $\epsilon$. \takashi{Is there any good citation for this?}\mor{I don't know.}  
Looking ahead, we need pseudorandomness against QPT distinguishers with quantum advice for our construction of quantum auxiliary-input commitments, and thus we first need to resolve this problem. 

Our key observation is that we can rely on recent progress on quantum random oracle model with quantum auxiliary-input~\cite{AC:HhaXagYam19,ITC:ChuLiaQia20,FOCS:CGLQ20,EC:Liu23} to resolve the issue. 
In particular, Liu proved that a length doubling random function $H:\bit^{\lfloor \secp/2 \rfloor}\rightarrow \bit^{\secp}$ is a PRG against adversaries that take polynomial-size quantum advice $\rho_H$ that depends on $H$ and make polynomially many quantum queries to $H$, i.e., the adversary cannot distinguish $H(x)$ for uniformly random $x\gets \bit^{\lfloor \secp/2 \rfloor}$ from a uniformly random string $y\gets \bit^{\secp}$. 
In particular, we will only need the result in the setting where the adversary does not make any query to $H$. 
By an averaging argument, this implies that for any fixed (possibly unbounded-time) adversary $\cA$,  
 there exists $H^*$ such that $H^*(x)$ for uniformly random $x\gets \bit^{\lfloor \secp/2 \rfloor}$ is indistinguishable from a uniformly random string $y\gets \bit^{\secp}$ for $\cA$ with any polynomial-size quantum advice. 
In particular, if $\cA$ is a quantum universal Turing machine, we can encode any non-uniform QPT computation into the quantum advice of $\cA$.   
Thus, by using the corresponding function $H^*$, we can see that the distribution of $H^*(x)$ for $x\gets\bit^{\lfloor \secp/2 \rfloor}$ is computationally indistinguishable from the uniformly random distribution over $\bit^{\secp}$ against any non-uniform QPT distinguishers with quantum advice. 
Moreover, it is clear that the distribution is sparse since the size of the support of the distribution is at most $2^{\lfloor \secp/2\rfloor}$ whereas the size of the whole space is $2^{\secp}$. 
This means that the distribution is a post-quantum sparse pseudorandom distribution. 
\takashi{The above explanation might be sloppy. Can we improve this?}\barak{I agree. I will think about if there is a good way to reword this.}

\paragraph{\bf Quantum auxiliary-input EFI pairs.}
Our next step is rather conceptual: We regard post-quantum sparse pseudorandom distributions as giving us an instance of quantum auxiliary-input EFI pairs. Here, quantum auxiliary-input EFI pairs are defined similarly to EFI pairs except that we do not require the states to be efficiently generatable but only require that the states are polynomial-size.\footnote{This is equivalent to requiring  that the states can be generated by a non-uniform QPT algorithm with quantum advice since any polynomial-size quantum state can be generated by the trivial 
algorithm that takes the state itself as quantum advice and simply outputs it.} That is, they are pairs of two (not necessarily efficiently generatable) quantum states that are statistically far but computationally indistinguishable. 
If $\{D_\secp\}_{\secp\in \mathbb{N}}$ is a post-quantum sparse pseudorandom distribution, then the following two states form a quantum auxiliary-input EFI pair:
\begin{align*}
    \xi_{\secp,0}\coloneqq\sum_{y\in \bit^\secp}D_\secp(y)\ket{y}\bra{y},~~~\xi_{\secp,1}\coloneqq\frac{1}{2^\secp}\sum_{y\in \bit^\secp}\ket{y}\bra{y}.
\end{align*}
Indeed, they are statistically far by the sparseness and
computationally indistinguishable by the post-quantum pseudorandomness. 

\paragraph{\bf Quantum auxiliary-input commitments.}
We then convert quantum auxiliary-input EFI pairs into quantum auxiliary-input commitments. 
In the plain model (where there is no auxiliary input), it is known that EFI pairs imply quantum commitments (and in fact, they are equivalent) as folllows~\cite{ITCS:BraCanQia23}. 
Let $(\xi_0,\xi_1)$ be an EFI pair. We can assume that the fidelity between $\xi_0$ and $\xi_1$ is $\negl(\secp)$ without loss of generality since otherwise we can exponentially decrease the fidelity by taking many copies of the state. 
For $b\in \bit$, 
let $\ket{\psi_b}$ be a purification of $\xi_b$ over registers $\qreg{X}$ and $\qreg{Y}$ where $\qreg{X}$ is the original register for $\xi_b$ and $\qreg{Y}$ is the register for purification. 
Then, a quantum commitment scheme is constructed as follows: to commit to a bit $b$, the committer generates $\ket{\psi_b}$ and  
sends register $\qreg{X}$ to the receiver as a commitment. 
In the reveal phase, the committer sends the corresponding purificaton register $\qreg{Y}$ along with the revealed bit $b$. 
Then the receiver applies a projective measurement $\{\ket{\psi_b}\bra{\psi_b},I-\ket{\psi_b}\bra{\psi_b}\}$ on the state in registers $(\qreg{X},\qreg{Y})$ and accepts if the first outcome is obtained, i.e., the state is projected onto $\ket{\psi_b}$.  
Then the hiding and binding of the protocol follows from computational indistinguishability and statistical farness of the EFI pair, respectively.

Our idea is to apply a similar construction in the quantum auxiliary-input setting. 
However, it does not directly work. First, if $(\xi_0,\xi_1)$ is a quantum auxiliary-input EFI pair, then $\ket{\psi_b}$ may not be efficiently generatable, and thus the committing procedure may not be efficiently done.\footnote{In our particular construction of quantum auxiliary-input EFI pairs, $\ket{\psi_1}$ is efficiently generatable, but $\ket{\psi_0}$ is not. 
We consider a more general setting where neither of them is efficiently generatable. 
\takashi{I added a footnote.}
} To resolve this, we can include $(\ket{\psi_0},\ket{\psi_1})$ in the quantum 
auxiliary input of the commitment scheme so that the committer can make use of it. 

The second issue is that the projective measurement $\{\ket{\psi_b}\bra{\psi_b},I-\ket{\psi_b}\bra{\psi_b}\}$  may not be efficiently realized because again, $\ket{\psi_b}$ may not be efficiently generatable. 
To resolve this issue, our idea is to let the receiver perform the SWAP test between the state sent from the committer and another fresh copy of $\ket{\psi_b}$ given as the receiver's quantum auxiliary input and accept if the test accepts.  
Note that the SWAP test between two states $\rho$ and $\sigma$ is accepted with probability 
$\frac{1+\Tr(\rho\sigma)}{2}$. Thus, the SWAP test by the receiver somewhat checks if the state sent from the committer is close to  $\ket{\psi_b}$ or not. However, this is not sufficient: malicious committers who sent other states can be accepted with probability at least 1/2, which breaks any meaningful notion of binding. 
Nevertheless, we can amplify this by running
the above protocol $m=\omega(\log \secp)$ times in parallel by including $m$ copies of  $(\ket{\psi_0},\ket{\psi_1})$ in the quantum auxiliary input. We can expect the binding error to be reduced to $2^{-m}=\negl(\secp)$ while preserving the hiding property by a standard hybrid argument.

\paragraph{\bf On definition of binding.}
While the proof of computational hiding is straightforward, proving binding  
is technically not as straightforward as one would expect. In the beginning, it is already unclear how we should define binding in the quantum auxiliary-input setting.  
The previous work~\cite{CKR16} adopted the notion that is often called \emph{sum-binding} and has been traditionally used in many works on quantum commitments~\cite{EC:DumMaySal00,EC:CreLegSal01,KO09,KO11,C:MorYam22}.\footnote{To our knowledge, the term ``sum-binding'' was used in \cite{EC:Unruh16} for the first time.} Sum-binding requires that after finishing the commit phase, 
if we let $p_b$ be the probability that the committer can reveal the commitment to $b$ and that is accepted by the receiver for $b\in \bit$, then we have $p_0+p_1\le 1+\negl(\secp)$. While this ensures some non-trivial security, Unruh~\cite{EC:Unruh16} pointed out that it is not sufficient if we want to use the quantum commitment as a building block of other cryptographic primitives such as zero-knowledge proofs and OTs.\footnote{In the plain model (where there is no quantum auxiliary-input), it is known that any (possibly interactive) quantum commitments can be compiled into a non-interactive ``canonical form''~\cite{AC:Yan22} and sum-binding for canonical form quantum commitments is sufficient for those applications~\cite{AC:Yan22,C:AnaQiaYue22,TQC:DAS23}. However, this does not seem to work in the quantum auxiliary-input setting.} 
For this reason, we adopt an extractor-based definition of binding introduced in~\cite{C:AnaQiaYue22}. 
It very roughly requires that there is a possibly inefficient extractor that extracts the committed bit from a commitment. This definition enables us to define a \emph{classical} bit committed in a quantum commitment, which is quite useful in security proofs of other protocols. 
In \Cref{sec:QAI}, we show that our construction satisfies the extractor-based binding by using an analysis of parallelly repeated SWAP tests~\cite{JACM:HarMon13} 
along with a similar template to those used in the proof of the extractor-based binding in the plain model~\cite{C:AnaQiaYue22,AC:Yan22}.

\paragraph{\bf Application to zero-knowledge proofs.}
To demonstrate applicability of our quantum auxiliary-input commitments, we use them to instantiate Blum's Hamiltonicity protocol~\cite{Blu87}.
This yields an unconditional construction of computational zero-knowledge proofs for $\mathbf{NP}$ with soundness error $1/2$  in the quantum auxiliary-input setting, i.e., the honest prover and honest verifier are allowed to take a common quantum auxiliary input.   
The proof of soundness is straightforward once we assume the extractor-based binding for the commitment: It ensures that we can ``extract'' the classical committed bits from commitments, after which the analysis is identical to the classical case. Note that we are proving statistical soundness and thus it does not matter that the extractor is inefficient. 
On the other hand, there is some subtlety in the proof, and actually also in the definition, of the zero-knowledge property. 
Recall that the standard quantum zero-knowledge property~\cite{SIAM:Watrous09} requires that there is a simulator that efficiently simulates the malicious verifier's view by using the same (possibly quantum) auxiliary-input as the verifier. 
We argue that this is unlikely to hold if 
we instantiate Blum's protocol using our quantum auxiliary-input commitments. 
In the protocol, the prover sends commitments to bits derived from the witness of the statement being proven and later reveals some of them. 
In our quantum auxiliary-input commitment scheme, after sending a commitment and then revealing the committed bit, the receiver gets $m$ copies of state $\ket{\psi_b}$ that is included in the quantum auxiliary input.  
Since $\ket{\psi_b}$ is not necessarily efficiently generatable, this already gives some ``knowledge'' to the verifier, i.e., 
the simulator cannot efficiently simulate $\ket{\psi_b}$. 
We believe that the above is just a definitional issue: Since the honest prover and receiver use the quantum auxiliary input, it is reasonable to allow the simulator to use them too.  In particular, we define the zero-knowledge property by using a \emph{non-uniform} simulator whose quantum auxiliary input can depend on the malicious verifier's quantum auxiliary input.\footnote{In this definition, the simulator is fully non-uniform, i.e., it can take arbitrary polynomial-size quantum auxiliary input. A more conservative definition would be to allow the simulator to only take the quantum auxiliary inputs of the malicious verifier and the honest prover, but we do not know how to prove such a variant.} 
While this is weaker than the conventional definition of zero-knowledge, it still ensures meaningful security, in particular, it implies witness indistinguishability. 

Defining the zero-knowledge property as above allows the proof to follow through. Remark that the simulator's quantum advice can include arbitrarily many copies of the quantum auxiliary input of the malicious verifier and the honest prover. This allows the simulator to circumvent the rewinding issue and instead use fresh quantum auxiliary input whenever it runs the prover and verifier. 


\paragraph{\bf Commitments in the CRQS model.}
Next, we explain how to modify our construction to one in the CRQS model. 
If we concretely write down our quantum auxiliary-input commitment scheme, 
its quantum auxiliary-input consists of $m$ copies of 
$\ket{\psi_{0}}=\sum_{x\in \bit^\secp}\ket{H(x)}\ket{x||0^\secp}$ and $\ket{\psi_{1}}=\sum_{y\in \bit^{2\secp}}\ket{y}\ket{y}$ 
where $H:\bit^\secp\rightarrow \bit^{2\secp}$ is an appropriately chosen function.\mor{Should we make it $H:\bit^{\secp/2}\to\bit^\secp$ for the consistency with other parts?} \takashi{I think it's fine.}
Since $H$ may not be efficiently computable, the quantum auxiliary input may not be efficiently generatable. 
However, we observe that we could actually use uniformly random $H$ if we allow the quantum auxiliary input to be randomized.  
Thus, our idea is to use many copies of $\ket{\psi_0}$ and $\ket{\psi_1}$ for a uniformly random $H$ as a CRQS.\footnote{Though $\ket{\psi_1}$ is efficiently generatable, we include it in the CRQS for the ease of presentation.} 
Since the CRQS must be efficiently generatable, we have to make sure that $H$ is efficiently computable. 
While this is not possible if $H$ is a completely random function, 
it is known that a $2q$-wise independent function is perfectly indistinguishable from a uniformly random function if the function is given as an oracle and the distinguisher is allowed to make at most $q$ quantum queries~\cite{C:Zhandry12}. Therefore, as long as the number of copies of the CRQS that is given to the adversary is at most $t$, then we can simulate the random function in quantum polynomial-time by using a $2m(t+1)$-wise independent function, which is computable in polynomial-time for any polynomial $t$.\footnote{Note that we should use $2m(t+1)$-wise independent function instead of $2mt$ since we need to also consider one copy of the CRQS that is used by the honest committer.}
This is the idea for our construction in the CRQS model.   
Remarkably, we can achieve statistical hiding instead of computational hiding unlike the quantum auxiliary-input setting. 
This crucially relies on the fact that the random classical key $k$ is hidden from the adversary, in particular, the malicious receiver (see also the next paragraph).  
So we can achieve both statistical hiding and statistical binding simultaneously in the CRQS model.

\paragraph{\bf Why the standard no-go does not apply in the CRQS model.} 
In the plain model, it is well-known that it is impossible to simultaneously satisfy both statistical hiding and statistical binding~\cite{LoChau97,Mayers97}, and such a result can also be
extended to the quantum auxiliary-input setting.
However, in the CRQS model, this no-go theorem does not apply.

To see why this is the case, consider an attempt to apply the argument of the no-go theorem to the CRQS model.
In the CRQS model, the setup algorithm samples a hidden key $k\gets\mathcal{K}_\secp$ and sends $\ket{\psi_k}$ to the committer and receiver.
Assume that the honest committer wants to commit a bit $b\in\bit$.
On input the CRQS, $\ket{\psi_k}$, the committer generates a state
$\ket{\Psi(\psi_k,b)}_{\qreg{R},\qreg{C}}$ over two registers 
$\qreg{R}$ and $\qreg{C}$, and sends 
$\qreg{C}$ to the receiver.
The receiver's state is then
\begin{align}
\sigma_b\coloneqq\frac{1}{|\mathcal{K}_\secp|}\sum_{k\in\mathcal{K}_\secp}
\Tr_{\qreg{R}}[\ket{\Psi(\psi_k,b)}\bra{\Psi(\psi_k,b)}_{\qreg{R},\qreg{C_1}}]
\otimes\ket{\psi_k}\bra{\psi_k}_{\qreg{C_2}}.
\end{align}
Assume that the commitment scheme satisfies statistical hiding. Then $\TD(\sigma_0,\sigma_1)=\negl(\secp)$.
However, due to Uhlmann's theorem~\cite{uhlmann1976transition}, there exists a unitary on the registers that purify $\sigma_b$ 
(in this case, the registers $\qreg{S}\cup\qreg{R}$ of the setup algorithm and the committer)
that transforms the purified version of $\sigma_0$
\begin{align}
\frac{1}{\sqrt{|\mathcal{K}_\secp|}}
\sum_{k\in\mathcal{K}_\secp}
\ket{k}_{\qreg{S}}\otimes
\ket{\Psi(\psi_k,0)}_{\qreg{R},\qreg{C_1}}
\otimes \ket{\psi_k}_{\qreg{C_2}}
\end{align}
to a state that is negligibly close to a purified version of $\sigma_1$.
However, notice that unlike in the plain model, this does not necessarily mean that the unbounded committer can break binding, because the committer does not have access to the register $\qreg{S}$. In other words, the committer does not have the ability to modify the classical key sampled by the setup algorithm.
It is for this reason that the prohibition against statistically binding, statistically hiding commitments fails to hold in this setting.
\barak{Should we also mention here the no-go for the classical CRS model?}
\takashi{either is fine for me as that's already mentioned in the introduction.}

\paragraph{\bf Applications in the CRQS model.} 
Finally, we argue applications of our quantum commitments in the CRQS model. 
We observe that we can simply plug our scheme into the compiler of \cite{C:BCKM21b} to obtain an OT in the CRQS model. 
We remark that we can apply this compiler unlike in the quantum auxiliary-input setting because the CRQS $\ket{\psi_k}$ is efficiently generatable given the classical key $k$, which can be generated by the simulator itself.

\subsection{Concurrent Work}
Qian \cite{Qian23} concurrently and independently shows similar results to ours. 
In particular, he also gives unconditionally secure quantum auxiliary-input commitments and quantum commitments in the CRQS model (with stateful setup) using our terminology. 
As additional results that are not present in our work,  
he also discusses how to generate the quantum auxiliary-input by exponential-time preprocessing and shows a barrier for constructing a classical analog.  

After exchanging our manuscripts and having some discussions, he came up with an idea of 
achieving $\epsilon$-simulation security using the compiler of \cite{C:BCKM21b}. 
This is presented in \cite[Appendix B]{Qian23}.

\section{Preliminaries}

\subsection{Basic notations}
We use standard notations of quantum computing and cryptography. $\secp$ is the security parameter.
$\negl$ is a negligible function.
$\poly$ is a polynomial.
For any set $A$, $x\gets A$ means that an element $x$ of $A$ is sampled from $A$ uniformly at random.
For any algorithm $A$, $y\gets A(x)$ means that the algorithm $A$ outputs $y$ on input $x$.
For sets $\calX$ and $\calY$, 
$\mathsf{Func}(\calX,\calY)$ is the set consisting of the all functions from $\calX$ to $\calY$. 
$\ket{\pm}$ means $\frac{1}{\sqrt{2}}(\ket{0} \pm \ket{1})$.
For the notational simplicity, we often omit the normalization factor of quantum states.
(For example, we often denote 
$\frac{1}{\sqrt{2}}(\ket{0}+\ket{1})$
by $\ket{0}+\ket{1}$.)
For any quantum state $\rho_{\qreg{A},\qreg{B}}$ over the registers $\qreg{A}$ and $\qreg{B}$,
$\Tr_{\qreg{A}}(\rho_{\qreg{A},\qreg{B}})$ is the partial trace over the register $\qreg{A}$.
For any quantum states $\rho$ and $\sigma$, 
$F(\rho,\sigma):=\Big(\mbox{Tr}\sqrt{\sqrt{\sigma}\rho\sqrt{\sigma}}\Big)^2$
is the fidelity,
and $\TD(\rho,\sigma)\coloneqq\frac{1}{2}\|\rho-\sigma\|_1$
is the trace distance.
For a pure state $\ket{\psi}$, we write $\|\ket{\psi}\|$ to mean the Euclidean norm of $\ket{\psi}$. 
$I\coloneqq\ket{0}\bra{0}+\ket{1}\bra{1}$ is the two-dimensional identity operator.
For simplicity, we often write $I^{\otimes m}$ just as $I$ when the dimension is clear from the context.
We write QPT to mean quantum polynomial time. A non-uniform QPT algorithm $\cA$ is specified by a family $\{\cA_\secp,\rho_\secp\}_{\secp\in \mathbb{N}}$ where $\cA_\secp$ and $\rho_\secp$ are a quantum circuit of size $\poly(\secp)$ and quantum advice of size $\poly(\secp)$ that are used when the input length is $\secp$, respectively.

\subsection{Lemmas}
We review several lemmas that are proven in existing works.

It is known that random oracles work as pseudorandom generators (PRGs) against adversaries that make quantum queries and take quantum advice. 
\begin{lemma}[\cite{EC:Liu23}]\label{thm:non-uniform_PRG}
    Let $N$, $M$, $T$, and $S$ be positive integers. 
    Let $\{\sigma_H\}_H$ be a family of $S$-qubit states indexed by $H:[N]\ra [M]$. 
    For any oracle-aided algorithm $\A$ that makes at most $T$ quantum queries to $H$, it holds that
    \begin{align*}
        \left|\Pr[1\gets\A^H(\sigma_H,y_0)]-\Pr[1\gets\A^H(\sigma_H,y_1)]\right|\leq 8\sqrt{2}\cdot \min_{\gamma>0}\left\{\left(\frac{S(T+1)}{\gamma N}+\frac{T^2}{N}\right)^{\frac{1}{2}}+\gamma\right\}
    \end{align*}
where $H\la \mathsf{Func}([N],[M]), x\la [N], y_0:=H(x)$, and $y_1\la [M]$. 
\end{lemma} 

\begin{corollary}\label{cor:QROM_PRG}
    Let $N$, $M$, and $S$ be positive integers. 
    Let $\{\sigma_H\}_H$ be a family of $S$-qubit states indexed by $H:[N]\ra [M]$. 
    For any algorithm $\A$ that does not make any query, it holds that
    \begin{align*}
        \left|\Pr[1\gets\A(\sigma_H,y_0)]-\Pr[1\gets\A(\sigma_H,y_1)]\right|\leq 16\sqrt{2}\cdot \left(\frac{S}{N}\right)^{\frac{1}{3}}
    \end{align*}
where $H\la \mathsf{Func}([N],[M]), x\la [N], y_0:=H(x)$, and $y_1\la [M]$. 
\end{corollary}
\begin{proof}
    Set $T:=0$ and $\gamma:=\left(\frac{S}{N}\right)^{\frac{1}{3}}$ in \Cref{thm:non-uniform_PRG}.
\end{proof}

It is known that $2q$-wise independent functions are indistinguishable from random functions by at most $q$ quantum queries.  
\begin{lemma}[\cite{C:Zhandry12}]\label{lem:simulation_QRO}
For any sets $\calX$ and $\calY$ of classical strings and $q$-quantum-query algorithm $\A$, we have
\[
\Pr[1\gets\A^{H}:H\gets \mathsf{Func}(\calX,\calY)]= \Pr[1\gets\A^{H}:H\gets \mathcal{H}_{2q}]
\]
where  
$\mathcal{H}_{2q}$ is a family of $2q$-wise independent functions from $\calX$ to $\calY$.
\end{lemma}

\begin{lemma}[{\cite[Lemma 2]{JACM:HarMon13}}]
\label{lem:HM13}
Let $\rho$ be a quantum state over $m$ registers $\qreg{A_1},...,\qreg{A_m}$.
Let $\sigma$ be a quantum state over $m$ registers $\qreg{B_1},...,\qreg{B_m}$.
Let us consider the following test.
\begin{enumerate}
    \item 
    For each $i\in[m]$, do the SWAP test\footnote{The SWAP test for two states $\rho_0$ and $\rho_1$ 
    is the following algorithm: On input $\ket{+}\bra{+}\otimes\rho_0\otimes\rho_1$, apply the controlled-SWAP gate so that the first qubit is the controlled qubit, and
    then measure the first qubit in the Hadamard basis. If the output is $\ket{+}$, the test is successful.} between $\qreg{A_i}$ and $\qreg{B_i}$.
    \item 
    Accept if all SWAP tests are successful.
\end{enumerate}
Then the probability that the above test accepts is
$\frac{1}{2^m}\sum_{S\subseteq[m]}\Tr[\rho_S\sigma_S]$,    
where $\rho_S$ is the state obtained by tracing out all $\qreg{A_i}$ of $\rho$ such that $i\notin S$, 
and
$\sigma_S$ is the state obtained by tracing out all $\qreg{B_i}$ of $\sigma$ such that $i\notin S$.
\end{lemma}

\begin{lemma}[{\cite[Lemma 31]{cryptoeprint:2020/1488}}]\label{lem:good_measurement}
Let $\ket{\psi_0}_{\qreg{X},\qreg{Y}}$ and $\ket{\psi_1}_{\qreg{X},\qreg{Y}}$ be pure states over registers $(\qreg{X},\qreg{Y})$ such that 
$$F(\Tr_{\qreg{Y}}(\ket{\psi_0}\bra{\psi_0}_{\qreg{X},\qreg{Y}}),\Tr_{\qreg{Y}}(\ket{\psi_1}\bra{\psi_1}_{\qreg{X},\qreg{Y}}))=\epsilon.$$
Then there is a projective measurement $\{\Pi_0,\Pi_1,\Pi_\bot\coloneqq I-\Pi_0-\Pi_1\}$ over register $\qreg{X}$ such that for each $b\in \bit$,  
$$
\left\|({\Pi_b}_{\qreg{X}}\otimes I_{\qreg{Y}})\ket{\psi_b}_{\qreg{X},\qreg{Y}}\right\|^2\ge 1-\sqrt{2\epsilon}.
$$
\end{lemma}

\begin{lemma}[{\cite[Lemma 4]{Rastegin_2007}}]\label{lem:TD_subnormalized}
    Let $\rho$ and $\sigma$ be any possibly subnormalized density operators. Then it holds that 
    \begin{align*}
    \TD(\rho,\sigma)=\max_{\Pi}\Tr(\Pi(\rho-\sigma))-\frac{\Tr(\rho-\sigma)}{2}
    \end{align*}
    where the maximum is taken over all projections $\Pi$.
\end{lemma}

\section{Post-Quantum Sparse Pseudorandom Distributions}
Goldreich and Krawczyk~\cite{DBLP:journals/rsa/GoldreichK92} introduced the notion of sparse pseudorandom distributions. 
A sparse pseudorandom distribution is a distribution supported by a sparse subset but looks uniformly random against non-uniform classical polynomial-time distinguishers. 
They showed that sparse pseudorandom distributions unconditionally exist.
Here, we introduce its \emph{post-quantum} version where the distribution looks uniformly random even against non-uniform QPT distinguishers. 
Then we show that post-quantum sparse pseudorandom distributions also exist unconditionally. 

\begin{definition}[Post-quantum sparse pseudorandom distributions]
We say that an ensemble $\{D_\secp\}_{\secp\in \mathbb{N}}$ of probabilistic distributions where $D_\secp$ is a distribution on $\bit^\secp$ 
is a post-quantum sparse pseudorandom distribution if the following hold:
\begin{description}
\item[Sparseness:]
$D_\secp$ is sparse, i.e., $\frac{|\mathsf{supp}(D_\secp)|}{2^\secp}=\negl(\secp)$.
\item[(Non-uniform) post-quantum pseudorandomness:] 
$\{D_\secp\}_{n\in \mathbb{N}}$ is computationally indistinguishable from the uniform distribution for non-uniform QPT distinguishers. 
That is, for any non-uniform QPT distinguisher $\{\A_\secp,\rho_\secp\}_{\secp\in \mathbb{N}}$, 
\begin{align*}
    \left|\Pr_{y_\secp\gets D_\secp}[\A_\secp(\rho_\secp,y_\secp)=1]-\Pr_{y'_\secp\gets \bit^{\secp}}[\A_\secp(\rho_\secp,y'_\secp)=1]\right|=\negl(\secp).
\end{align*}
\end{description}
\end{definition}

\begin{theorem}
\label{thm:pq-sparse_pseudorandom}
There exist post-quantum sparse pseudorandom distributions.
\end{theorem}
\begin{proof} \takashi{revised the proof.}
For any ensemble $\cH\coloneqq\{H_\secp\}_{\secp\in \mathbb{N}}$ of functions $H_\secp:\bit^{\lfloor \secp/2\rfloor} \rightarrow \bit^\secp$, let $D_{\mathcal{H}}:=\{D_{\mathcal{H},\secp}\}_{\secp\in \mathbb{N}}$ where $D_{\mathcal{H},\secp}$ is the distribution of $H_\secp(x)$ for $x\gets \bit^{\lfloor \secp/2\rfloor}$. Since the image size of $H_\secp$ is at most $2^{\lfloor \secp/2\rfloor}$,  $D_{\mathcal{H}}$ is an ensemble of sparse distributions for all $\mathcal{H}$. For completing the proof, it suffices to prove that there exists $\mathcal{H}$ such that $D_{\mathcal{H}}$ is pseudorandom. 
For each $\secp\in \mathbb{N}$, 
function $H_\secp:\bit^{\lfloor \secp/2\rfloor} \rightarrow \bit^\secp$, and a quantum distinguisher $\A_{\secp}$ with quantum advice $\rho_{\secp}$, let $\mathsf{Adv}[H_\secp,\A_{\secp},\rho_{\secp}]$ be the distinguisher's advantage, i.e., 
\begin{align*}
   \mathsf{Adv}[H_\secp,\A_{\secp},\rho_{\secp}]:= \left|\Pr_{x_\secp\gets \bit^{\lfloor \secp/2 \rfloor}}[\A_\secp(\rho_\secp,H_\secp(x_\secp))=1]-\Pr_{y_\secp\gets \bit^{\secp}}[\A_\secp(\rho_\secp,y_\secp)=1]\right|.
\end{align*}
For each $\secp\in \mathbb{N}$, 
let $\mathcal{S}_{\secp}$ be the set consisting of all tuples $(\A_{\secp},\rho_{\secp})$ of a distinguisher and its quantum advice such that the sum of the description size of $\A_{\secp}$ (as a bit string) and the length of $\rho_{\secp}$ is at most $2^{\secp/4}$. 
For each $\secp\in \mathbb{N}$, let $H^*_\secp$ be the ``most pseudorandom'' function against distinguishers in $\mathcal{S}_{\secp}$, i.e., 
\begin{align*}
H^*_\secp:= \underset{H_\secp}{\operatorname{argmin}} 
\underset{(\A_{\secp},\rho_{\secp})\in \mathcal{S}_{\secp}}{\max}
\mathsf{Adv}[H_\secp,\A_{\secp},\rho_{\secp}],
\end{align*}
where $\operatorname{argmin}$ is taken over 
$H_\secp\in \mathsf{Func}(\bit^{\lfloor \secp/2 \rfloor},\bit^\secp)$ 
and let 
\begin{align*}
\epsilon^*(\secp):=\underset{(\A_{\secp},\rho_{\secp})\in \mathcal{S}_{\secp}}{\max}
\mathsf{Adv}[H^*_\secp,\A_{\secp},\rho_{\secp}].
\end{align*}
Below, we prove that $\epsilon^*(\secp)=\negl(\secp)$. 
If this is proven, it is easy to see that $D_{\mathcal{H}^*}$ for $\mathcal{H}^*:=\{H^*_\secp\}_{\secp\in \mathbb{N}}$ is pseudorandom noting that for any non-uniform QPT distinguisher $\{\A_{\secp},\rho_{\secp}\}_{\secp\in \mathbb{N}}$, we have $(\A_{\secp},\rho_{\secp})\in \mathcal{S}_{\secp}$ for sufficiently large $\secp$. 
Thus, we are left to prove $\epsilon^*(\secp)=\negl(\secp)$.

By the definitions of $H^*_\secp$ and $\epsilon^*(\secp)$, for any function $H_\secp \in \mathsf{Func}(\bit^{\lfloor \secp/2 \rfloor},\bit^\secp)$, there is $(\A_{H_\secp},\rho_{H_\secp})\in \mathcal{S}_\secp$ such that 
\begin{align}\label{eq:lowerbound_advantage}
\mathsf{Adv}[H_\secp,\A_{H_\secp},\rho_{H_\secp}]\ge \epsilon^*(\secp).
\end{align}
For each $\secp$, 
we consider an algorithm $\B_\secp$ and a family $\{\sigma_{H_\secp}\}_{H_\secp}$ of quantum advice indexed by functions $H_\secp:\bit^{\lfloor \secp/2\rfloor} \rightarrow \bit^\secp$ that tries to break the inequality of \Cref{cor:QROM_PRG} for $N=2^{\lfloor \secp/2\rfloor}$ and $M=2^\secp$ as follows.\footnote{We identify $[N]$ (resp. $[M]$) with $\bit^{\lfloor \secp/2\rfloor}$ (resp. $\bit^\secp$) in a natural way.}
\begin{itemize}
    \item The advice $\sigma_{H_\secp}$ is defined to be a tuple of $(\A_{H_\secp},\rho_{H_\secp})$
    and an additional bit $b$ that takes $0$ if  $\Pr_{x_\secp\gets \bit^{\lfloor \secp/2 \rfloor}}[\A_\secp(\rho_\secp,H_\secp(x_\secp))=1]-\Pr_{y_\secp\gets \bit^{\secp}}[\A_\secp(\rho_\secp,y_\secp)=1]\ge 0$ and otherwise takes $1$. 
    Note that the size $S$ of $\sigma_{H_\secp}$ is at most $2^{\secp/4}+1$ since $(\A_{H_\secp},\rho_{H_\secp})\in \mathcal{S}_\secp$. 
    \item $\B_\secp$ takes the advice $\sigma_{H_\secp}=(\A_{H_\secp},\rho_{H_\secp})$ and an input $y\in \bit^\secp$, which is $y=H_\secp(x)$ for $x\gets\bit^{\lfloor\secp/2\rfloor}$ or uniformly sampled from $\bit^\secp$, and runs $\A_{H_\secp}$ on advice $\rho_{H_\secp}$ and input $y$.
    If $b=0$, 
    $\B_\secp$ outputs 
    whatever $\A_{H_\secp}$ outputs and if $b=1$, $\B_\secp$ flips the output of $\A_{H_\secp}$  and outputs it.
\end{itemize}
It is easy to see that $\B_\secp$'s distinguishing advantage between the two cases is
\begin{align*}
\underset{H_\secp}{\mathbb{E}} \left[\mathsf{Adv}[H_\secp,\A_{H_\secp},\rho_{H_\secp}]\right] \ge \epsilon^*(\secp)
\end{align*}
where  
the inequality holds because \Cref{eq:lowerbound_advantage} holds for all $H_\secp$. 
On the other hand, \Cref{cor:QROM_PRG} implies that the distinguishing advantage is at most $16\sqrt{2}\left(\frac{S}{N}\right)$, which is negligible in $\secp$. 
This implies that $\epsilon^*(\secp)$ is negligible in $\secp$. 
\end{proof}
\begin{remark}
    It was pointed out by Fermi Ma that an alternative conceptually simpler proof of \Cref{thm:pq-sparse_pseudorandom} with a better bound is possible using the matrix Chernoff bound based on a similar idea to that used in \cite[Section 2.2]{cryptoeprint:2023/1602}. 
\end{remark}
\section{Quantum Auxiliary-Input EFI Pairs}
Brakerski, Canetti, and Qian~\cite{ITCS:BraCanQia23} introduced the notion of EFI pairs as a pair of two mixed states that are efficiently generatable, statistically far, and computationally indistinguishable.
Here, we consider its quantum auxiliary-input version where the efficiently generatable requirement is omitted. 
We show that post-quantum sparse pseudorandom distributions imply quantum auxiliary-input EFI pairs.
From \cref{thm:pq-sparse_pseudorandom},
it means that quantum auxiliary-input EFI pairs unconditionally exist.  
\begin{definition}[Quantum auxiliary-input EFI pairs]
A quantum auxiliary-input EFI pair is a family $\{\xi_{\secp,b}\}_{\secp\in\mathbb{N},b\in\bit}$ of (mixed) states that satisfies the following:
\begin{itemize}
\item[\bf Polynomial size:] For each $b\in \bit$, $\xi_{\secp,b}$ is a $\poly(\secp)$-qubit state. 
\item[\bf Statistically far:] $\TD(\xi_{\secp,0},\xi_{\secp,1})=1-\negl(\secp)$.
\item[\bf Computational indistinguishability:] 
$\xi_{\secp,0}$ and $\xi_{\secp,1}$ are computationally indistinguishable against non-uniform QPT distinguishers. That is, for any non-uniform QPT distinguisher $\{\A_\secp,\rho_\secp\}_{\secp\in \mathbb{N}}$, 
\begin{align*}
    \left|\Pr[\A_\secp(\rho_\secp,\xi_{\secp,0})=1]-\Pr[\A_\secp(\rho_\secp,\xi_{\secp,1})=1]\right|=\negl(\secp).
\end{align*}
\end{itemize}
\end{definition}
\begin{remark}
In the original definition of EFI pairs in \cite{ITCS:BraCanQia23}, the statistically far property only requires that $\TD(\xi_{\secp,0},\xi_{\secp,1})=1/\poly(\secp)$. However, this can be generically amplified to $1-\negl(\secp)$ by parallel repetition while preserving the computational indistinguishability. We note that we can use a hybrid argument to show that parallel repetition preserves computational indistinguishability even though $\xi_{\secp,b}$ is not efficiently generatable 
since we require computational indistinguishability against non-uniform distinguishes with quantum advice. 
\end{remark}

\begin{theorem}
    Quantum auxiliary-input EFI pairs exist.
\end{theorem}
\begin{proof}
Let $\{D_\secp\}_{\secp\in \mathbb{N}}$ be a post-quantum sparse pseudorandom distribution, which exists by \Cref{thm:pq-sparse_pseudorandom}.  
We define
\begin{align*}
    &\xi_{\secp,0}:=\sum_{y\in \bit^\secp}D_\secp(y)\ket{y}\bra{y}\\
    &\xi_{\secp,1}:=\sum_{y\in \bit^\secp}2^{-\secp}\ket{y}\bra{y}.
\end{align*}
They satisfy the polynomial size property since they are $\secp$-qubit states. 
To see the statistically far property, consider the following unbounded-size distinguisher: Given  $\xi_{\secp,0}$ or $\xi_{\secp,1}$, measure the state in the computational basis and output $0$ if the measurement outcome belongs to $|\mathsf{supp}(D_\secp)|$ and otherwise outputs $1$. We can see that its distinguishing advantage is $1-\frac{|\mathsf{supp}(D_\secp)|}{2^\secp}=1-\negl(\secp)$ by the sparseness of $\{D_\secp\}_{\secp\in \mathbb{N}}$. This implies $\TD(\xi_{\secp,0},\xi_{\secp,1})=1-\negl(\secp)$. Finally, computational indistinguishability immediately follows from post-quantum pseudorandomness of $\{D_\secp\}_{\secp\in \mathbb{N}}$. 
\end{proof}
\section{Quantum Auxiliary-Input Commitments}\label{sec:QAI}
In this section, we define quantum auxiliary-input commitments following \cite{CKR16}, and construct it without relying on any assumption. 
\subsection{Definition}
We define quantum auxiliary-input commitments following \cite{CKR16}. 
\begin{definition}[Quantum auxiliary-input commitments~\cite{CKR16}]
\label{def:QAI}
A (non-interactive) quantum auxiliary-input commitment scheme is given by a tuple of QPT algorithms $C$ (a committer), $R$ (a receiver), and  a family $\{\ket{\psi_\secp}\}_{\secp\in \mathbb{N}}$ of $\poly(\secp)$-qubit states (referred to as quantum auxiliary inputs).  
The scheme is divided into the following three phases, the quantum auxiliary-input phase,
the commit phase, and the reveal phase.   
\begin{itemize}
\item {\bf Quantum auxiliary-input phase:} 
For the security parameter $\secp$, a single copy of $\ket{\psi_\secp}$ is sent to $C$, and a single copy of $\ket{\psi_\secp}$ is sent to $R$. 
    \item
{\bf Commit phase:} 
$C$ takes a bit $b\in \bit$ to commit (and $\ket{\psi_\secp}$ given in the first phase) as input, generates a quantum state on registers $\qreg{C}$ and $\qreg{R}$, and 
sends the register $\qreg{C}$ to $R$. 
\item
{\bf Reveal phase:}
$C$ sends $b$ and the register $\qreg{R}$ to $R$.
$R$ takes $(b,\qreg{C},\qreg{R})$ given by $C$ (and $\ket{\psi_\secp}$ given in the first phase) as input, and outputs $b$ if it accepts and otherwise outputs $\bot$. 
\end{itemize}
As correctness, we require that $R$ accepts with probability $1$ if the protocol is run honestly.
\end{definition}
Below, we define security of quantum auxiliary-input commitments. First, we define hiding and sum-binding  following \cite{CKR16}.  
\begin{definition}[Hiding]\label{def:QAI_hiding}
A quantum auxiliary-input commitment scheme $(C,R,\{\ket{\psi_\secp}\}_{\secp\in \mathbb{N}})$ satisfies statistical (resp. computational) hiding if for any non-uniform unbounded-time (resp. QPT) algorithm $\{\A_\secp,\rho_\secp\}_{\secp\in \mathbb{N}}$, 
\begin{align*}
\left|
    \Pr[1\gets\A_\secp(\rho_\secp,\Tr_{\qreg{R}}(\sigma_{\qreg{C},\qreg{R}})):\sigma_{\qreg{C},\qreg{R}} \la C_{\mathsf{com}}(\ket{\psi_\secp},0)]
    -\right.\\
    \left.
    \Pr[1\gets\A_\secp(\rho_\secp,\Tr_{\qreg{R}}(\sigma_{\qreg{C},\qreg{R}})):\sigma_{\qreg{C},\qreg{R}} \la C_{\mathsf{com}}(\ket{\psi_\secp},1)]
    \right|\le \negl(\secp)
\end{align*}
where $C_{\mathsf{com}}$ is the commit phase of $C$. 
\end{definition}
\begin{remark}
In the above definition, we do not explicitly include $\ket{\psi_\secp}$ as input to $\A_\secp$. However, it holds even if $\A_\secp$ is given polynomially many copies of $\ket{\psi_\secp}$ since they could have been included in the quantum advice $\rho_\secp$. 
\end{remark}

\begin{definition}[Sum-binding]\label{def:QAI_binding}
A quantum auxiliary-input commitment scheme $(C,R,\{\ket{\psi_\secp}\}_{\secp\in \mathbb{N}})$ satisfies statistical (resp. computational)  sum-binding if the following holds. For any pair of non-uniform unbounded-time (resp. QPT) malicious committers $C^*_0$ and $C^*_1$ that take $\ket{\psi_\secp}$ as a quantum auxiliary input and
work in the same way in the commit phase, if we let $p_b$ to be the probability that $R$ accepts the revealed bit $b$ in the interaction with $C^*_b$ for $b\in \bit$, then we have 
\begin{align*}
p_0+p_1\le 1+\negl(\secp).
\end{align*}
\end{definition}

While the sum-binding is the notion used in the previous work~\cite{CKR16}, we introduce a stronger definition which we call extractor-based binding inspired by \cite{C:AnaQiaYue22}. 
The motivation of introducing this definition is that it is more suitable for applications like zero-knowledge proofs. 

\begin{definition}[Extractor-based binding]\label{def:QAI_extract}
A quantum auxiliary-input commitment scheme $(C,R,\{\ket{\psi_\secp}\}_{\secp\in \mathbb{N}})$ satisfies statistical (computational) extractor-based binding if there is a non-uniform unbounded-time algorithm $\mathcal{E}=\{\mathcal{E}_\secp\}_{\secp \in \mathbb{N}}$ (called the extractor) such that for any non-uniform unbounded-time (resp. QPT) malicious committer $C^*$, 
$\mathsf{Real}_\secp^{C^*}$ and $\mathsf{Ideal}_\secp^{C^*,\mathcal{E}}$ are indistinguishable against non-uniform unbounded-time (resp. QPT) distinguishers 
where the experiments $\mathsf{Real}_\secp^{C^*}$ and   $\mathsf{Ideal}_\secp^{C^*,\mathcal{E}}$ are defined as follows.
\begin{itemize}
\item $\mathsf{Real}_\secp^{C^*}$: 
The malicious committer 
$C^*$ takes $\ket{\psi_\secp}$ as input, and interacts with the honest receiver $R$ in the commit and reveal phases. 
Let $b\in \{0,1,\bot\}$ be the output of $R$ and $\tau_{C^*}$ be the final state of $C^*$. The experiment outputs a tuple $(\tau_{C^*},b)$. 
\item $\mathsf{Ideal}_\secp^{C^*,\mathcal{E}}$: 
$C^*$ takes $\ket{\psi_\secp}$ as input, and runs its commit phase to generate a commitment $\sigma_{\qreg{C}}$ where $C^*$ may keep a state that is entangled with $\sigma_{\qreg{C}}$. 
$\mathcal{E}_\secp$ takes the register $\qreg{C}$ as input, outputs an extracted bit $b^*\in \bit$, and sends a post-execution state on $\qreg{C}$ (that may be different from the original one) to $R$ as a commitment.  
Then $C^*$ and $R$ run the reveal phase.
Let $b$ be the output of $R$ and $\tau_{C^*}$ be the final state of $C^*$. 
If $b\notin \{\bot,b^*\}$, then the experiment outputs a special symbol $\fail$ and otherwise outputs a tuple $(\tau_{C^*},b)$. 
\end{itemize}
\end{definition}

We can easily see the following lemma.
\begin{lemma}\label{lem:QAI_ext_to_sum}
Statistical (resp. computational) extractor-based binding implies statistical (resp. computational) sum-binding. 
\end{lemma}
\begin{proof}
    Suppose that there are non-uniform unbounded-time (resp. QPT) malicious committers $C^*_0$ and $C^*_1$ that break statistical (resp. computational) sum-binding.  
    This implies that
    \begin{align*}
    \Pr[b_0=0:(\tau_{C^*_0},b_0) \gets \mathsf{Real}_\secp^{C^*_0}]+\Pr[b_1=1:(\tau_{C^*_1},b_1) \gets \mathsf{Real}_\secp^{C^*_1}]- 1
    \end{align*}
    is non-negligible. Here, $C^*_0$ and $C^*_1$ work in the same way in the commit phase.
    By statistical (resp. computational) extractor-based binding, the above implies that 
     \begin{align}
    \Pr[b_0=0:(\tau_{C^*_0},b_0) \gets \mathsf{Ideal}_\secp^{C^*_0,\mathcal{E}}]+\Pr[b_1=1:(\tau_{C^*_1},b_1) \gets \mathsf{Ideal}_\secp^{C^*_1,\mathcal{E}}]- 1
    \label{nonnegligible}
    \end{align}
    is non-negligible. 
    However, since $C^*_0$ and $C^*_1$ share the same commit phase and the extracted bit $b^*$ only depends on the commit phase, we clearly have 
     \begin{align*}
    \Pr[b_0=0:(\tau_{C^*_0},b_0) \gets \mathsf{Ideal}_\secp^{C^*_0,\mathcal{E}}]+\Pr[b_1=1:(\tau_{C^*_1},b_1) \gets \mathsf{Ideal}_\secp^{C^*_1,\mathcal{E}}]\le 1.
    \end{align*}
    This contradicts the fact that \cref{nonnegligible} is non-neglibible. Thus, statistical (resp. computational) sum-binding is satisfied. 
\end{proof}

\subsection{Construction}\label{sec:construction_QAI}
We show that quantum auxiliary-input commitments can be constructed without any assumption.

Let $\{\xi_{\secp,b}\}_{\secp\in\mathbb{N},b\in\bit}$ be a quantum auxiliary-input EFI pair,  
which exists unconditionally.
Let $\ket{\psi_{\secp,b}}_{\qreg{X},\qreg{Y}}$ be a purification of $\xi_{\secp,b}$ such that $\Tr_{\qreg{Y}}(\ket{\psi_{\secp,b}}\bra{\psi_{\secp,b}}_{\qreg{X},\qreg{Y}})=\xi_{\secp,b}$.  
Then we construct a quantum auxiliary-input commitment scheme as follows where $m=\secp$.\footnote{In fact, it suffices to set $m=\omega(\log \secp)$. }
\begin{itemize}
    \item 
    {\bf Quantum auxiliary-input phase:}
    For security parameter $\secp\in \mathbb{N}$, define the quantum auxiliary input 
    $\ket{\psi_\secp}:=\ket{\psi_{\secp,0}}^{\otimes m}\otimes \ket{\psi_{\secp,1}}^{\otimes m}$.
    One single copy of $\ket{\psi_\secp}$ is sent to $C$ and the other single copy of $\ket{\psi_\secp}$ is sent to $R$.
    \item
{\bf Commit phase:} 
$C$ takes 
a bit $b\in \bit$ to commit 
(and the quantum auxiliary input $\ket{\psi_\secp}$ given in the first phase)  
as input. 
$C$ sends the register $\qreg{C}$ of $\ket{\psi_{\secp,b}}^{\otimes m}_{\qreg{C},\qreg{R}}$ to $R$,
where $\qreg{C}:=(\qreg{X_1},...,\qreg{X_m})$ and  $\qreg{R}:=(\qreg{Y_{1}},...,\qreg{Y_{m}})$.\footnote{Here, the $\qreg{X}$ ($\qreg{Y}$) register of
the $i$th $\ket{\psi_{\secp,b}}$ of $\ket{\psi_{\secp,b}}^{\otimes m}$ is set to $\qreg{X_i}$ ($\qreg{Y_i}$).}
\item
{\bf Reveal phase:}
$C$ sends $b$ and the register $\qreg{R}$ to $R$.
In the first phase, $R$ received $\ket{\psi_\secp}$. $R$ uses only
$\ket{\psi_{\secp,b}}^{\otimes m}_{\qreg{C'},\qreg{R'}}$, where
$\qreg{C'}:=(\qreg{X_1'},...,\qreg{X_m'})$ and  $\qreg{R'}:=(\qreg{Y_1'},...,\qreg{Y_m'})$.
    For each $i\in[m]$, 
     $R$ runs the SWAP test between registers 
     $(\qreg{X_{i}},\qreg{Y_{i}})$ and 
     $(\qreg{X'_{i}},\qreg{Y'_{i}})$. 
      If all of the tests accept, $R$ accepts by outputting $b$ and otherwise rejects by outputting $\bot$.
\end{itemize}
If the protocol is run honestly, the registers $(\qreg{X_i},\qreg{Y_i})$ and 
     $(\qreg{X'_i},\qreg{Y'_i})$ take exactly the same state $\ket{\psi_{\secp,b}}$ for all $i\in[m]$. 
     Thus, $R$ accepts with probability $1$ and thus correctness holds. 
     
\begin{theorem}\label{thm:QAI_hid_bin}
The above protocol satisfies computational hiding and statistical extractor-based binding.  
\end{theorem}
\begin{proof}~
\paragraph{\bf Computational hiding.}
It immediately follows from computational indistinguishability of the quantum auxiliary-input EFI pair.

\paragraph{\bf Statistical extractor-based binding.}
Since $\TD(\xi_{\secp,0},\xi_{\secp,1})=1-\negl(\secp)$ we have $F(\xi_{\secp,0},\xi_{\secp,1})=\negl(\secp)$.
Let $\epsilon:=F(\xi_{\secp,0},\xi_{\secp,1})$. 
Then \Cref{lem:good_measurement} implies that 
there is a projective measuremt $\{\Pi_0,\Pi_1,\Pi_\bot=I-\Pi_0-\Pi_1\}$ on register $\qreg{X}$ such that 
for each $b\in \bit$, 
it holds that
\begin{align}\label{eq:I-Pi_b}
\left\|\left(I-\Pi_b\right)\ket{\psi_{\secp,b}}\right\|^2\le\sqrt{2\epsilon}.
\end{align}
Here, we simply write $\Pi_b$ to mean ${\Pi_b}_\qreg{X}\otimes I_{\qreg{Y}}$ for simplicity.

The extractor $\mathcal{E}_\secp$ is described as follows:
Upon receiving the commitment register $\qreg{C}=(\qreg{X_1},...,\qreg{X_m})$, 
for each $i\in [m]$, 
it applies the projective measurement $\{\Pi_0,\Pi_1,\Pi_\bot\}$ on $\qreg{X_i}$ to obtain an outcome $b_i\in \{0,1,\bot\}$. 
\begin{itemize}
\item If $|\{i\in[m]:b_i=0\}|>2m/3$, it outputs $b^*=0$.
\item  If $|\{i\in[m]:b_i=1\}|>2m/3$, it outputs $b^*=1$.
\item Otherwise, it outputs $b^*=\bot$.
\end{itemize} 

By using an additional ancilla register $\qreg{E}$, we can describe the extractor $\mathcal{E}_\secp$ as a projective measurement $\{\widetilde{\Pi}_0,\widetilde{\Pi}_1,\widetilde{\Pi}_\bot\}$ over registers $\qreg{C}$ and $\qreg{E}$, i.e., they are defined in such a way that running  $\mathcal{E}_\secp$ on 
$\rho_\qreg{C}$ is equivalent to applying the projective measurement $\{\widetilde{\Pi}_0,\widetilde{\Pi}_1,\widetilde{\Pi}_\bot\}$ on $\rho_{\qreg{C}}\otimes \ket{0...0}\bra{0...0}_{\qreg{E}}$ for any state $\rho_\qreg{C}$. 

For $b\in \bit$, let $P_b$ be the projector corresponding to the purified verification procedure by the receiver $R$. That is, it is defined so that $R$ is described as follows: 
Upon receiving the quantum auxiliary input $\ket{\psi_{\secp,b}}^{\otimes m}_\qreg{Aux'}$ where $\qreg{Aux'}:=\{\qreg{X'_i},\qreg{Y'_i}\}_{i\in [m]}$ and $(b,\qreg{C},\qreg{R})$ from the committer, $R$ initializes its ancilla register $\qreg{V}$ to  $\ket{0...0}_{\qreg{V}}$, applies the projective measurement $(P_b,I-P_b)$ on registers $(\qreg{Aux'},\qreg{C},\qreg{R},\qreg{V})$, and accepts (i.e., outputs $b$) if the measurement outcome is $P_b$ (i.e., the state is projected onto the image of $P_b$) and otherwise rejects (i.e., outputs $\bot$). 

Let $C^*$ be an unbounded-time malicious committer. By using the above notations, we can describe the experiments 
$\mathsf{Real}_\secp^{C^*}$ and $\mathsf{Ideal}_\secp^{C^*,\mathcal{E}}$ in the following form where we highlight the differences in red bold texts:

\begin{description}
\item[$\mathsf{Real}_\secp^{C^*}$:]~
\begin{enumerate}
\item $C^*$  takes the quantum auxiliary input $\ket{\psi_\secp}$ as input and  generates a state $\ket{\phi_b}$ over registers $(\qreg{C},\qreg{R},\qreg{W})$ along with a bit $b\in \bit$ where $\qreg{W}$ is the internal register of $C^*$.\footnote{It may be possible that $C^*$ generates $b$ after applying some unitary in the reveal phase. However, since it does not receive anything between the commit and reveal phases, we can assume that $C^*$ does not apply any unitary between the two phases without loss of generality.} 
\item $C^*$ sends $\qreg{C}$ to $R$.
\item 
$C^*$ sends $(b,\qreg{R})$ to $R$. 
\item 
$R$ runs the verification in the reveal phase. That is, $R$ takes the quantum auxiliary input $\ket{\psi_{\secp,b}}^{\otimes m}_\qreg{Aux'}$, $(b,\qreg{C},\qreg{R})$ sent from $C^*$, and the ancilla qubits 
$\ket{0...0}_{\qreg{V}}$, applies the projective measurement $(P_b,I-P_b)$ on registers $(\qreg{Aux'},\qreg{C},\qreg{R},\qreg{V})$. 
\item  If $R$ rejects (i.e., the measurement results in projection onto the image of $I-P_b$), $b$ is replaced by $\bot$.  
\item The experiment outputs $(b,\qreg{W})$.
\end{enumerate}
\end{description}

\begin{description}
\item[$\mathsf{Ideal}_\secp^{C^*,\mathcal{E}}$:]~
\begin{enumerate}
\item $C^*$  takes the quantum auxiliary input $\ket{\psi_\secp}$ as input and  generates a state $\ket{\phi_b}$ over registers $(\qreg{C},\qreg{R},\qreg{W})$ along with a bit $b\in \bit$ where $\qreg{W}$ is the internal register of $C^*$. 
\item $C^*$ sends $\qreg{C}$ to $R$.
\item \textcolor{red}{\bf The extractor $\mathcal{E}_\secp$ 
prepares the ancilla qubits $\ket{0...0}_{\qreg{E}}$, 
applies the projective measurement $\{\widetilde{\Pi}_0,\widetilde{\Pi}_1,\widetilde{\Pi}_\bot\}$ on registers $\qreg{C}$ and $\qreg{E}$, lets $b^*\in\{0,1,\bot\}$ be the outcome, and sends $\qreg{C}$ to $R$.}   
\item $C^*$ sends $(b,\qreg{R})$ to $R$. 
\item $R$ runs the verification in the reveal phase. That is, $R$ takes the quantum auxiliary input $\ket{\psi_{\secp,b}}^{\otimes m}_\qreg{Aux'}$, $(b,\qreg{C},\qreg{R})$ sent from $C^*$ and $\mathcal{E}$, and the ancilla qubits $\ket{0...0}_{\qreg{V}}$, applies the projective measurement $(P_b,I-P_b)$ on registers $(\qreg{Aux'},\qreg{C},\qreg{R},\qreg{V})$. 
\item \textcolor{red}{\bf If $R$ accepts (i.e., the measurement results in projection onto the image of $P_b$) and $b\ne b^*$, the experiment outputs the special symbol $\fail$ and halts.} 
\item If $R$ rejects (i.e., the measurement results in projection onto the image of $I-P_b$), $b$ is replaced by $\bot$. 
\item The experiment outputs $(b,\qreg{W})$.
\end{enumerate}
\end{description} 

The rest of the proof is similar to \cite[Appendix B]{cryptoeprint:2020/1488} where it is proven that canonical quantum bit commitments satisfy the extractor-based binding property defined in \cite{C:AnaQiaYue22}. 
We define density operators $\rho_{\mathsf{real}}(b)$ and $\rho_{\mathsf{ideal}}(b)$ that correspond to the outputs of $\mathsf{Real}_\secp^{C^*}$ and $\mathsf{Ideal}_\secp^{C^*,\mathcal{E}}$ conditioned on the bit sent from $C^*$ being $b$, respectively. 
That is, we let  
\begin{align*}
 \rho_{\mathsf{real}}(b):= \Tr_{\qreg{Aux'},\qreg{C},\qreg{R},\qreg{V},\qreg{E}}(P_b \ket{\widetilde{\phi}_b}\bra{\widetilde{\phi}_b})\otimes \ket{b}\bra{b} +   \Tr_{\qreg{Aux'},\qreg{C},\qreg{R},\qreg{V},\qreg{E}}((I-P_b) \ket{\widetilde{\phi}_b}\bra{\widetilde{\phi}_b})\otimes \ket{\bot}\bra{\bot} 
\end{align*}
and 
\begin{align*}
 \rho_{\mathsf{ideal}}(b):= &\Tr_{\qreg{Aux'},\qreg{C},\qreg{R},\qreg{V},\qreg{E}}(N_b \ket{\widetilde{\phi}_b}\bra{\widetilde{\phi}_b})\otimes \ket{b}\bra{b} +   \Tr_{\qreg{Aux'},\qreg{C},\qreg{R},\qreg{V},\qreg{E}}(N_\bot \ket{\widetilde{\phi}_b}\bra{\widetilde{\phi}_b})\otimes \ket{\bot}\bra{\bot}\\
 &+ \Tr_{\qreg{Aux'},\qreg{C},\qreg{R},\qreg{V},\qreg{E}}(N_{\fail} \ket{\widetilde{\phi}_b}\bra{\widetilde{\phi}_b})\otimes \ket{\fail}\bra{\fail}
\end{align*}
where 
\begin{align*}
    \ket{\widetilde{\phi}_b}:=\ket{\phi_b}_{\qreg{C},\qreg{R},\qreg{W}}\otimes \ket{\psi_{\secp,b}}^{\otimes m}_{\qreg{Aux'}}\otimes \ket{0...0}_{\qreg{V}}\otimes\ket{0...0}_{\qreg{E}}
\end{align*}
and 
\begin{align*}
N_b:=\widetilde{\Pi}_b P_b \widetilde{\Pi}_b,~~~ N_{\fail}:=(I-\widetilde{\Pi}_{b}) P_b (I-\widetilde{\Pi}_{b}),~~~ N_\bot:=I-N_b-N_{\fail}.
\end{align*}
We note that $N_{\fail}$ and $N_\bot$ depend on $b$, but we omit the dependence for simplicity. 
Noting that the distribution of $b$ is identical in the both experiments, 
it suffices to prove 
\begin{align*}
    \TD( \rho_{\mathsf{real}}(b), \rho_{\mathsf{ideal}}(b))=\negl(\secp) 
\end{align*}
for $b\in \bit$. 
We show the following lemma.
\begin{lemma}\label{lem:bound_fail}
For each $b\in \bit$,  it holds that
$$\|P_b (I-\widetilde{\Pi}_{b}) \ket{\widetilde{\phi}_b}\|^2=\negl(\secp).$$ 
\end{lemma}
\begin{proof}[Proof of \Cref{lem:bound_fail}]
By the definitions of $P_b$ and $\ket{\widetilde{\phi}_b}$,  
$\|P_b (I-\widetilde{\Pi}_{b}) \ket{\widetilde{\phi}_b}\|^2$ is the probability that
$\rho:=\Tr_{\qreg{W},\qreg{E}}\left((I-\widetilde{\Pi}_{b})_{\qreg{C},\qreg{E}}\ket{\phi_b}_{\qreg{C},\qreg{R},\qreg{W}}\ket{0...0}_{\qreg{E}}\right)$ passes the verification by $R$ with respect to the committed bit $b$.\footnote{Precisely speaking, it is not correct because $\rho$ is not a state (it is not normalized).
However, we believe the abuse of terminology will not cause any confusion.} 
By \Cref{lem:HM13}, it is equal to 
\begin{align*}
\frac{1}{2^m}\sum_{S\subseteq [m]}\Tr(\rho_S \sigma_S)
\end{align*}
where 
$\sigma:=\bigotimes_{i\in [m]} \ket{\psi_{\secp,b}}\bra{\psi_{\secp,b}}_{\qreg{X_i},\qreg{Y_i}}$. 
Here, $\rho_S$ and $\sigma_S$ are the states obtained by tracing out all $(\qreg{X_i},\qreg{Y_i})$ such that $i\notin S$.\footnote{Again,
$\rho_S$ is not a state because it is not normalized.} 
By the definition of $\widetilde{\Pi}_b$, we can write\footnote{The purified extractor is assumed to work as follows. It ``coherently
measures'' $\{\Pi_0,\Pi_1,\Pi_\bot\}$ on each $\qreg{X_i}$ and writes the result $b_1\|...\|b_m$ on an ancilla register. Depending on
the value on the ancilla register, 0, 1, or $\bot$ is written in another ancilla register. Finally the second register is measured in
the computational basis. If we trace out those two registers, we get \cref{postmeasurement} as the (sub-normalized) post-measurement state.} 
\begin{align}
\rho=\sum_{T\subseteq [m]:|T|\le 2m/3} 
\Tr_{\qreg{W}}
\hat{\Pi}_{b}^{(T)}
\ket{\phi_b}\bra{\phi_b} 
\hat{\Pi}_{b}^{(T)} 
\label{postmeasurement}
\end{align}
where 
\begin{align*}
\hat{\Pi}_{b}^{(T)}
:=
\left(\bigotimes_{i\in T} \Pi_{b\qreg{X_i}}\right)\otimes 
\left(\bigotimes_{i\in [m]\setminus T}(I-\Pi_b)_{\qreg{X_i}}\right).
\end{align*}

Thus, we have 
\begin{align*}
\frac{1}{2^m}\sum_{S\subseteq [m]}\Tr(\rho_S \sigma_S)
=
\frac{1}{2^m}\sum_{S\subseteq [m]}
\sum_{\substack{T\subseteq [m]\\|T|\le 2m/3}} 
\Tr\left(
\left(\Tr_{\{\qreg{X_i},\qreg{Y_i}\}_{i\notin S},\qreg{W}}
\hat{\Pi}_{b}^{(T)}
\ket{\phi_b}\bra{\phi_b} 
\hat{\Pi}_{b}^{(T)} \right)
\left(\bigotimes_{i\in S}\ket{\psi_{\secp,b}}\bra{\psi_{\secp,b}}_{\qreg{X_i},\qreg{Y_i}}\right)\right).
\end{align*}
For any $S$ and $T$, 
we have 
\begin{align*}
&\Tr\left(
\left(\Tr_{\{\qreg{X_i},\qreg{Y_i}\}_{i\notin S},\qreg{W}}
\hat{\Pi}_{b}^{(T)}
\ket{\phi_b}\bra{\phi_b} 
\hat{\Pi}_{b}^{(T)}  
\right)
\left(\bigotimes_{i\in S}\ket{\psi_{\secp,b}}\bra{\psi_{\secp,b}}_{\qreg{X_i},\qreg{Y_i}}\right)\right)\\ 
=&\Tr\left(
(I-\Pi_b)^{\otimes |S\setminus T|}_{\{\qreg{X_{i}}\}_{i\in S \setminus T}}
\left(\Tr_{\{\qreg{X_i},\qreg{Y_i}\}_{i\notin S},\qreg{W}}
\hat{\Pi}_{b}^{(T)}
\ket{\phi_b}\bra{\phi_b} 
\hat{\Pi}_{b}^{(T)} \right)
(I-\Pi_b)^{\otimes |S\setminus T|}_{\{\qreg{X_{i}}\}_{i\in S \setminus T}}
\left(\bigotimes_{i\in S}\ket{\psi_{\secp,b}}\bra{\psi_{\secp,b}}_{\qreg{X_i},\qreg{Y_i}}\right)\right)\\
=&\Tr\left(
\left(\Tr_{\{\qreg{X_i},\qreg{Y_i}\}_{i\notin S},\qreg{W}}
\hat{\Pi}_{b}^{(T)}
\ket{\phi_b}\bra{\phi_b} 
\hat{\Pi}_{b}^{(T)} \right)
(I-\Pi_b)^{\otimes |S\setminus T|}_{\{\qreg{X_{i}}\}_{i\in S \setminus T}}
\left(\bigotimes_{i\in S}\ket{\psi_{\secp,b}}\bra{\psi_{\secp,b}}_{\qreg{X_i},\qreg{Y_i}}\right)
(I-\Pi_b)^{\otimes |S\setminus T|}_{\{\qreg{X_{i}}\}_{i\in S \setminus T}}
\right)\\
\le &
\left\|\hat{\Pi}_{b}^{(T)}
\ket{\phi_b}\right\|^2 \cdot\Tr\left(
(I-\Pi_b)^{\otimes |S\setminus T|}_{\{\qreg{X_{i}}\}_{i\in S \setminus T}}
\left(\bigotimes_{i\in S}\ket{\psi_{\secp,b}}\bra{\psi_{\secp,b}}_{\qreg{X_i},\qreg{Y_i}}\right)
(I-\Pi_b)^{\otimes |S\setminus T|}_{\{\qreg{X_{i}}\}_{i\in S \setminus T}}
\right)\\
\leq &\left\|\hat{\Pi}_{b}^{(T)}
\ket{\phi_b}\right\|^2 \cdot\left\|(I-\Pi_b)\ket{\psi_{\secp,b}}\right\|^{2|S\setminus T|}\\ 
\le &\left\|\hat{\Pi}_{b}^{(T)}
\ket{\phi_b}\right\|^2 \cdot(2\epsilon)^{|S\setminus T|}
\end{align*}
where 
the first inequality follows from $\Tr_{\{\qreg{X_i},\qreg{Y_i}\}_{i\notin S},\qreg{W}}
\hat{\Pi}_{b}^{(T)}
\ket{\phi_b}\bra{\phi_b} 
\hat{\Pi}_{b}^{(T)} \le \left\|\hat{\Pi}_{b}^{(T)}
\ket{\phi_b}\right\|^2 \cdot I$ 
and the final inequality follows from \Cref{eq:I-Pi_b}. 

Thus, for any $T$ such that $|T|\le 2m/3$, we have 
\begin{align*}
&\frac{1}{2^m}\sum_{S\subseteq [m]}
\Tr\left(
\left(\Tr_{\{\qreg{X_i},\qreg{Y_i}\}_{i\notin S},\qreg{W}}
\hat{\Pi}_{b}^{(T)}
\ket{\phi_b}\bra{\phi_b} 
\hat{\Pi}_{b}^{(T)} \right)
\left(\bigotimes_{i\in S}\ket{\psi_{\secp,b}}\bra{\psi_{\secp,b}}_{\qreg{X_i},\qreg{Y_i}}\right)\right)\\
\le&
\frac{1}{2^m}\sum_{S\subseteq [m]}
\left\|\hat{\Pi}_{b}^{(T)}
\ket{\phi_b}\right\|^2 \cdot(2\epsilon)^{|S\setminus T|}\\
\le &
\left\|\hat{\Pi}_{b}^{(T)}
\ket{\phi_b}\right\|^2 \cdot
\frac{1}{2^m}\left(\sum_{S\subseteq [m]:S\subseteq T} 1 + \sum_{S\subseteq [m]:S\nsubseteq T} 2\epsilon\right)\\
\le & \left\|\hat{\Pi}_{b}^{(T)}
\ket{\phi_b}\right\|^2 \cdot \left(2^{-m/3}+ 2\epsilon \right)
\end{align*}
where the final inequality follows from 
$|\{S\subseteq [m]:S\subseteq T\}|=2^{|T|}\leq 2^{2m/3}$ and 
$|\{S\subseteq [m]:S\nsubseteq T\}|\le 2^m$. 

Therefore, we have 
\begin{align*}
\frac{1}{2^m}\sum_{S\subseteq [m]}\Tr(\rho_S \sigma_S)
&\le 
\sum_{T\subseteq [m]:|T|\le 2m/3} 
\left\|\hat{\Pi}_{b}^{(T)}
\ket{\phi_b}\right\|^2 \cdot \left(2^{-m/3}+ 2\epsilon \right)\\
&\le 2^{-m/3}+ 2\epsilon \le \negl(\secp)
\end{align*}
where the final inequality follows from $m=\omega(\log \secp)$ and $\epsilon=\negl(\secp)$. 
This completes the proof of \Cref{lem:bound_fail}.

\end{proof}

We define 
\begin{align*}
&\tau_{\mathsf{real}}^{(b)}:=\Tr_{\qreg{Aux'},\qreg{C},\qreg{R},\qreg{V},\qreg{E}}(P_b \ket{\widetilde{\phi}_b}\bra{\widetilde{\phi}_b}),~~~ \tau_{\mathsf{real}}^{(\bot)}:=\Tr_{\qreg{Aux'},\qreg{C},\qreg{R},\qreg{V},\qreg{E}}((I-P_b) \ket{\widetilde{\phi}_b}\bra{\widetilde{\phi}_b}),\\
&\tau_{\mathsf{ideal}}^{(b)}:= \Tr_{\qreg{Aux'},\qreg{C},\qreg{R},\qreg{V},\qreg{E}}(N_b \ket{\widetilde{\phi}_b}\bra{\widetilde{\phi}_b}),~~~ \tau_{\mathsf{ideal}}^{(\bot)}:=\Tr_{\qreg{Aux'},\qreg{C},\qreg{R},\qreg{V},\qreg{E}}(N_\bot \ket{\widetilde{\phi}_b}\bra{\widetilde{\phi}_b}),\\ 
&\tau_{\mathsf{ideal}}^{(\fail)}:=\Tr_{\qreg{Aux'},\qreg{C},\qreg{R},\qreg{V},\qreg{E}}(N_{\fail} \ket{\widetilde{\phi}_b}\bra{\widetilde{\phi}_b}) 
\end{align*}
so that we can write 
\begin{align}\label{eq:rhoreal}
 \rho_{\mathsf{real}}(b)=\tau_{\mathsf{real}}^{(b)}\otimes \ket{b}\bra{b} +   \tau_{\mathsf{real}}^{(\bot)}\otimes \ket{\bot}\bra{\bot} 
\end{align}
and 
\begin{align}\label{eq:rhoideal}
 \rho_{\mathsf{ideal}}(b)= \tau_{\mathsf{ideal}}^{(b)}\otimes \ket{b}\bra{b} +  \tau_{\mathsf{ideal}}^{(\bot)}\otimes \ket{\bot}\bra{\bot}+\tau_{\mathsf{ideal}}^{(\fail)}\otimes \ket{\fail}\bra{\fail}.
\end{align}
Then we use \Cref{lem:bound_fail} to show the following lemma. 
\begin{lemma}\label{lem:TD_bound}
For each $b\in \bit$, the following hold:
\begin{enumerate}
    \item $\Tr(\tau_{\mathsf{ideal}}^{(\fail)})=\negl(\secp)$.
    \item $\TD(\tau_{\mathsf{real}}^{(b)},\tau_{\mathsf{ideal}}^{(b)})=\negl(\secp)$.
    \item $\TD(\tau_{\mathsf{real}}^{(\bot)},\tau_{\mathsf{ideal}}^{(\bot)})=\negl(\secp)$. 
\end{enumerate}
\end{lemma}
\begin{proof}[Proof of \Cref{lem:TD_bound}]~\\

    \noindent\textbf{First item.} 
   We have
    \begin{align*}
        \Tr(\tau_{\mathsf{ideal}}^{(\fail)})=\Tr(N_{\fail} \ket{\widetilde{\phi}_b}\bra{\widetilde{\phi}_b})=\bra{\widetilde{\phi}_b}(I-\widetilde{\Pi}_{b}) P_b (I-\widetilde{\Pi}_{b}) \ket{\widetilde{\phi}_b}= \|P_b (I-\widetilde{\Pi}_{b}) \ket{\widetilde{\phi}_b}\|^2=\negl(\secp)
    \end{align*}
    where the final equality follows from 
    \Cref{lem:bound_fail}. 
    
  \noindent\textbf{Second item.}  
     We have
    \begin{align*}
        \TD(\tau_{\mathsf{real}}^{(b)},\tau_{\mathsf{ideal}}^{(b)})
        & \le \TD(P_b \ket{\widetilde{\phi}_b},P_b\widetilde{\Pi}_b \ket{\widetilde{\phi}_b})\\
        &= \max_{Q} \Tr\left(Q\left(P_b \ket{\widetilde{\phi}_b}\bra{\widetilde{\phi}_b} P_b-P_b  \widetilde{\Pi}_b\ket{\widetilde{\phi}_b}\bra{\widetilde{\phi}_b}\widetilde{\Pi}_b P_b\right)\right)\\         
        &~~~-\frac{1}{2}\left(\Tr(P_b \ket{\widetilde{\phi}_b}\bra{\widetilde{\phi}_b} P_b)-\Tr(P_b \widetilde{\Pi}_b \ket{\widetilde{\phi}_b} \bra{\widetilde{\phi}_b}\widetilde{\Pi}_b P_b)\right)\\
        &=\max_{Q} \left(\|QP_b \ket{\widetilde{\phi}_b}\|^2-\|QP_b \widetilde{\Pi}_b \ket{\widetilde{\phi}_b}\|^2\right) 
        -\frac{1}{2}\left(\|P_b \ket{\widetilde{\phi}_b}\|^2 - \|P_b \widetilde{\Pi}_b \ket{\widetilde{\phi}_b}\|^2 \right),
    \end{align*}
    where the maximum is taken over all projectors $Q$ on $(\qreg{W},\qreg{Aux'},\qreg{C},\qreg{R},\qreg{V},\qreg{E})$,  
    the first inequality follows from the monotonicity of the trace norm, 
    and the first equality follows from \Cref{lem:TD_subnormalized}. 
For any projector $Q$, 
we have
\begin{align*}
    \left|\|QP_b \ket{\widetilde{\phi}_b}\|^2-\|QP_b \widetilde{\Pi}_b \ket{\widetilde{\phi}_b}\|^2\right|
    &=\left|\|QP_b (\widetilde{\Pi}_b+(I-\widetilde{\Pi}_{b})) \ket{\widetilde{\phi}_b}\|^2-\|QP_b \widetilde{\Pi}_b \ket{\widetilde{\phi}_b}\|^2\right|\\
    &=|\bra{\widetilde{\phi}_b}(I-\widetilde{\Pi}_{b}) P_b QP_b (\widetilde{\Pi}_b+(I-\widetilde{\Pi}_{b})) \ket{\widetilde{\phi}_b}+
    \bra{\widetilde{\phi}_b}\widetilde{\Pi}_{b} P_b QP_b (I-\widetilde{\Pi}_{b}) \ket{\widetilde{\phi}_b}
    |\\
    &\le \|P_b (I-\widetilde{\Pi}_{b}) \ket{\widetilde{\phi}_b}\|\cdot \|QP_b (\widetilde{\Pi}_b+(I-\widetilde{\Pi}_{b})) \ket{\widetilde{\phi}_b}\|\\
    &~~~+\|Q P_b \widetilde{\Pi}_{b} \ket{\widetilde{\phi}_b}\|\cdot \|P_b (I-\widetilde{\Pi}_{b}) \ket{\widetilde{\phi}_b}\|\\
    &\le 2\|P_b (I-\widetilde{\Pi}_{b}) \ket{\widetilde{\phi}_b}\|=\negl(\secp) 
\end{align*}
where the final equality follows from \Cref{lem:bound_fail}. 
In particular, the above also holds for the case of $Q=I$. 
Combining the above, 
we have  $\TD(\tau_{\mathsf{real}}^{(b)},\tau_{\mathsf{ideal}}^{(b)})=\negl(\secp)$.\\

    \noindent\textbf{Third item.}
    \begin{align*}
    \TD(\tau_{\mathsf{real}}^{(\bot)},\tau_{\mathsf{ideal}}^{(\bot)})
    &=\TD\left(\Tr_{\qreg{Aux'},\qreg{C},\qreg{R},\qreg{V},\qreg{E}}((I-P_b) \ket{\widetilde{\phi}_b}\bra{\widetilde{\phi}_b}),\Tr_{\qreg{Aux'},\qreg{C},\qreg{R},\qreg{V},\qreg{E}}((I-N_b-N_\fail) \ket{\widetilde{\phi}_b}\bra{\widetilde{\phi}_b})\right)\\
    &=\TD\left(\Tr_{\qreg{Aux'},\qreg{C},\qreg{R},\qreg{V},\qreg{E}}(P_b \ket{\widetilde{\phi}_b}\bra{\widetilde{\phi}_b}),\Tr_{\qreg{Aux'},\qreg{C},\qreg{R},\qreg{V},\qreg{E}}((N_b+N_\fail) \ket{\widetilde{\phi}_b}\bra{\widetilde{\phi}_b})\right)\\
    &\le \TD\left(\Tr_{\qreg{Aux'},\qreg{C},\qreg{R},\qreg{V},\qreg{E}}(P_b \ket{\widetilde{\phi}_b}\bra{\widetilde{\phi}_b}),\Tr_{\qreg{Aux'},\qreg{C},\qreg{R},\qreg{V},\qreg{E}}(N_b \ket{\widetilde{\phi}_b}\bra{\widetilde{\phi}_b})\right)+\frac{1}{2}\Tr(N_{\fail} \ket{\widetilde{\phi}_b}\bra{\widetilde{\phi}_b})\\
    &=\negl(\secp)
    \end{align*}
    where we have used the triangle inequality in the third line, and the first and second items of this lemma in the final line. 
    This completes the proof of \Cref{lem:TD_bound}. 
\end{proof}

By \Cref{eq:rhoreal,eq:rhoideal}, 
we have 
\begin{align*}
    \TD( \rho_{\mathsf{real}}(b), \rho_{\mathsf{ideal}}(b))=\TD(\tau_{\mathsf{real}}^{(b)},\tau_{\mathsf{real}}^{(b)})+\TD(\tau_{\mathsf{real}}^{(\bot)},\tau_{\mathsf{ideal}}^{(\bot)})+\frac{1}{2}\Tr(\tau_{\mathsf{ideal}}^{(\fail)})=\negl(\secp)
\end{align*}
where the final equality follows from \Cref{lem:TD_bound}.
This implies statistical binding and completes the proof of \Cref{thm:QAI_hid_bin}.
\end{proof}

\subsection{Impossibility of Statistical Security}
It is easy to prove that achieving both statistical hiding and statistical (sum)-binding is impossible even in the quantum auxiliary-input setting  based on the similar impossibility result in the plain model~\cite{LoChau97,Mayers97}. 
\begin{theorem}
    There do not exist statistically hiding and statistically sum-binding quantum auxiliary-input commitments. 
\end{theorem}
\begin{proof}[Proof sketch] 
Given a quantum auxiliary-input commitment scheme, we can regard the classical description of the quantum auxiliary-input as a classical auxiliary-input if we allow the sender and receiver to be unbounded-time.
Since the impossibility of \cite{LoChau97,Mayers97} is applicable to classical auxiliary-input quantum commitments with unbounded-time sender and receiver, the theorem follows. 
More concretely, assume that the honest committer generates $\ket{\Psi(\psi_\secp,b)}_{\qreg{R},\qreg{C}}$ over two registers $\qreg{R}$ and $\qreg{C}$ to commit to $b\in\bit$
given the quantum auxiliary input $\ket{\psi_\secp}$.
Due to the statistical hiding, $\Tr_{\qreg{R}}(\ket{\Psi(\psi_\secp,0)}_{\qreg{R},\qreg{C}})$
should be negligibly close to
$\Tr_{\qreg{R}}(\ket{\Psi(\psi_\secp,1)}_{\qreg{R},\qreg{C}})$.
However, it means, from the Uhlmann's theorem, that
there exists a unitary $U$ on $\qreg{R}$ that maps
$\ket{\Psi(\psi_\secp,0)}_{\qreg{R},\qreg{C}}$ to a state that is negligibly close to
$\ket{\Psi(\psi_\secp,1)}_{\qreg{R},\qreg{C}}$, which breaks the statistical binding.
\end{proof}

Thus, it is necessary to relax either of hiding or binding to the computational one. 
In \Cref{thm:QAI_hid_bin}, we unconditionally construct a computationally hiding and statistically binding construction.
On the other hand, we do not know how to construct statistically hiding and computationally binding one without assuming any assumption. 
We remark that it is unclear if the flavor conversion theorems in the plain model~\cite{EC:CreLegSal01,AC:Yan22,EC:HhaMorYam23} work in the quantum auxiliary-input setting.

\ifnum\submission=1
\subsection{Application to Zero-Knowledge Proofs}
To demonstrate the applicability of our quantum auxiliary-input commitments, in 
\ifnum\inclappndx=0
the full version of this paper~\cite{cryptoeprint:2023/1844}%
\else
\Cref{sec:zk}%
\fi
, we show an application to zero-knowledge proofs. We use quantum auxiliary-input commitments to instantiate Blum's Hamiltonicity protocol~\cite{Blu87}, yielding an unconditional construction of computational zero-knowledge proofs for $\mathbf{NP}$ with soundness error $1/2$ in the quantum auxiliary-input setting. See 
\ifnum\inclappndx=0
the full version~\cite{cryptoeprint:2023/1844}
\else
\Cref{sec:zk}
\fi
for details.
\fi

\ifnum\submission=0
    \section{Application to Zero-Knowledge Proofs}
\label{sec:zk}
To demonstrate the applicability of our quantum auxiliary-input commitments, we show an application to zero-knowledge proofs. Though it works for both $\mathbf{NP}$ and $\mathbf{QMA}$, we focus on the case of $\mathbf{NP}$ here for simplicity.

\subsection{Definition}
\paragraph{\bf Notation.}
For an $\NP$ language $\lang$ and $x\in \lang$, $\rel_{\lang}(x)$ is the set that consists of all (classical) witnesses $w$ such that the verification machine for $\lang$ accepts $(x,w)$.
For an interactive protocol between a prover $P$ and a verifier $V$, 
we denote by $\execution{\pro(x_{\pro})}{\ver(x_{\ver})}$ an execution of the protocol where $x_\pro$ is $\pro$'s input and $x_\ver$ is $\ver$'s input.
We denote by $\OUT_\ver\execution{\pro(x_{\pro})}{\ver(x_{\ver})}$ the final output of $\ver$ in the execution. 
An honest verifier's output is $\top$ indicating acceptance or $\bot$ indicating rejection, and a quantum malicious verifier's output may be an arbitrary quantum state.  

\begin{definition}[Quantum Auxiliary-Input Zero-Knowledge Proofs for $\NP$]
A quantum auxiliary-input zero-knowledge proof for an $\mathbf{NP}$ language $\lang$ is given by a tuple of QPT algorithms $P$ (a prover) and $V$ (a verifier) and a family $\{\ket{\psi_\secp}\}_{\secp\in \mathbb{N}}$ of $\poly(\secp)$-qubit states (referred to as quantum auxiliary inputs).  
$P$ and $V$ take a statement $x$ of $\lang$ as a common input and $P$ takes a witness $w$ for $x$ as its private input. 
In addition, each of $P$ and $V$ takes a copy of the quantum auxiliary input $\ket{\psi_{|x|}}$. 
Then $P$ and $V$ interact with each other through a quantum channel, and $V$ finally outputs $\top$ indicating acceptance or $\bot$ indicating rejection. 

We require the following properties:
\begin{itemize}
\item {\bf Perfect completeness.}
For any $x\in \lang$ and $w\in R_\lang(x)$, we have 
\begin{align*}
    \Pr[\OUT_V\execution{P(\ket{\psi_{|x|}},x,w)}{V(\ket{\psi_{|x|}},x)}=\top]=1.
\end{align*}
\item {\bf Statistical $s$-soundness.}
For any non-uniform unbounded-time cheating prover $\pro^*$, there exists a negligible function $\negl$ such that for any $\secp \in \mathbb{N}$ and any $x\in \bit^\secp\setminus \lang $, we have  
\begin{align*}
    \Pr[\OUT_\ver\execution{\pro^*(\ket{\psi_\secp},x)}{\ver(\ket{\psi_{\secp}},x)}=\top]\leq s(\secp)+\negl(\secp).
\end{align*}
\item {\bf Non-uniform computational quantum zero-knowledge.}
For any non-uniform QPT malicious verifier $V^*=\{V^*_\secp,\rho_\secp\}_{\secp\in \mathbb{N}}$, there is a non-uniform QPT simulator $\mathsf{Sim}=\{\mathsf{Sim}_\secp,\sigma_\secp\}_{\secp\in \mathbb{N}}$ and a negligible function $\negl$ such that for any $\secp\in \mathbb{N}$, $x\in \lang \cup \bit^\secp$, $w\in R_\lang(x)$, and non-uniform QPT distinguisher $\cA=\{\cA_\secp, \tau_\secp\}_{\secp\in \mathbb{N}}$, it holds that
\begin{align*}
    |\Pr[\cA_\secp(\tau_\secp,\OUT_{\ver^*_\secp}\execution{\pro(\ket{\psi_\secp},x,w)}{\ver^*_\secp(\rho_\secp,x)})=1]-
    \Pr[\cA_\secp(\tau_\secp,\mathsf{Sim}_\secp(\sigma_\secp,x))=1]|
    \le \negl(\secp). 
\end{align*}
\end{itemize}
\end{definition}

\begin{remark}\label{rem:sequential}
Unlike the standard definition of zero-knowledge, the above definition allows the simulator to take quantum advice that may depend on the malicious verifier's advice. In particular, this is crucial in our security proof since the simulator's quantum advice contains many copies of the malicious verifier's advice. A caveat of this definition is that it is not closed under polynomially many sequential compositions because the size of simulator's quantum advice may exponentially blow up. But it still ensures a meaningful notion of security. For example, a simple hybrid argument shows that it implies witness indistinguishability defined below:

\noindent\textbf{Computational witness indistinguishability.}
 For any $\secp\in \mathbb{N}$, $x\in \lang \cup \bit^\secp$, two witnesses $w_0,w_1\in R_\lang(x)$, and 
 non-uniform QPT malicious verifier $\cV=\{\cV_\secp, \rho_\secp\}_{\secp\in \mathbb{N}}$, it holds that 
\begin{align*}
    |\Pr[\OUT_{\ver^*_\secp}\execution{\pro(\ket{\psi_\secp},x,w_0)}{\ver^*_\secp(\rho_\secp,x)}=1]-
   \Pr[\OUT_{\ver^*_\secp}\execution{\pro(\ket{\psi_\secp},x,w_1)}{\ver^*_\secp(\rho_\secp,x)}=1]|
    \le \negl(\secp). 
\end{align*}
\end{remark}

\subsection{Construction}
We simply plug our quantum auxiliary-input commitment scheme into Blum's Hamiltonicity protocol~\cite{Blu87}. 
Here, we briefly recall basic notations for graphs and the graph Hamiltonicity problem.
A graph $G$ with vertices $W$ is represented as a set of its edges, i.e., for $(i,j)\in W^2$, $(i,j)\in G$ if and only if $G$ has an edge between vertices $i$ and $j$. 
An instance of the graph Hamiltonicity problem is an (undirected) graph $G$. 
A graph $G$ is an Yes instance (i.e., belongs to the corresponding language) if it has a Hamiltonian cycle, i.e., a path that goes through all the vertices of $G$ exactly once. The Hamiltonian cycle is the corresponding witness. 

We describe the protocol below:
Let $\Pi_{\mathsf{com}}$ be a quantum auxiliary-input commitment scheme.
Let $\{\ket{\psi_{\secp}}\}_{\secp\in \mathbb{N}}$ be the family of quantum auxiliary input for  $\Pi_{\mathsf{com}}$. 
Let $G$ be a graph (instance) with vertices $W$ and $H$ be its Hamiltonian cycle (witness). 
\begin{itemize}
\item {\bf Inputs.} 
$P$ takes $\ket{\psi_{\secp}}^{\otimes p(\secp)}$, $G$, and $H$ as input.
$V$ takes $\ket{\psi_{\secp}}^{\otimes p(\secp)}$ and $G$ as input.
Here, $\secp$ denotes the bit-length of the description of $G$, and $p(\secp)$ is the number of total bits that are committed by $P$ throughout the protocol. 
\item {\bf First round.}
$P$ picks a random permutation $\pi$ on $W$. 
Then $P$ commits to the adjency matrix of $\pi(G):=\{(\pi(i),\pi(j)):(i,j)\in G\}$ using $\Pi_{\mathsf{com}}$ in a bit-wise manner and sends the commitment registers to $V$
\item {\bf Second round.}
$V$ sends a random bit $c\in \bit$ to $P$.  
\item {\bf Third round.}
\begin{itemize}
\item If $c=0$, then $P$ sends $\pi$ and the reveal registers of the all commitments. 
\item If $c=1$, then $P$ sends the reveal registers to the commitments that correspond to $\pi(H):=\{(\pi(i),\pi(j)):(i,j)\in H\}$. 
\end{itemize}
\item {\bf Decision.}
\begin{itemize}
\item If $c=0$, $V$ verifies all the revealed commitments. 
If any of verification fails, $V$ outputs $\bot$.
Otherwise, $V$ recovers a graph $G'$ from the revealed bits. If $G'=\pi(G)$, then $V$ outputs $\top$ and otherwise outputs $\bot$.  
\item If $c=1$, $V$ verifies all the revealed commitments. If any of verification fails, $V$ outputs $\bot$.
Otherwise, $V$ recovers a graph $H'$  from the revealed bits. If $H$ is a Hamiltonian cycle, $V$ outputs $\bot$.
\end{itemize}
\end{itemize}

\begin{theorem}\label{thm:ZK}
The above protocol satisfies perfect completeness, $1/2$-soundness, and quantum auxiliary-input computational zero-knowledge.
\end{theorem}
\begin{remark}
Unfortunately, we do not know how to reduce the soundness error while keeping quantum auxiliary-input computational zero-knowledge since it is not preserved even under sequential composition as remarked in \Cref{rem:sequential}. 
On the other hand, since it implies witness indistinguishability that is preserved under parallel composition, we can obtain a protocol that satisfies perfect completeness, $\negl(\secp)$-soundness, and computational witness indistinguishability. 
\end{remark}
\begin{proof}[Proof of \cref{thm:ZK}]
The perfect completeness is clear from the description of the protocol. Below, we prove $1/2$-soundness and quantum auxiliary-input computational zero-knowledge.

\paragraph{\bf Soundness.}
Suppose that a graph $G$ does not have a Hamiltonian cycle.  
Suppose that there is a non-uniform unbounded prover $P^*$ that makes $V$ accept with $1/2+\mathsf{non}\text{-}\mathsf{negl}(\secp)$ for a false instance $G$ that does not have a Hamiltonian cycle where  $\mathsf{non}\text{-}\mathsf{negl}(\secp)$ is some non-negligible function.  
Let $\tilde{V}$ be a non-uniform unbounded-time verifier that works similarly to $V$ except for the following differences: 
\begin{itemize}
    \item It runs the extractor of $\Pi_{\mathsf{com}}$ after receiving the commitments in the first round, and recovers a graph $\tilde{G}$ 
    from the extracted bits. 
    \item If the revealed bits in the third round is inconsistent to $\tilde{G}$, then it immediately rejects. 
\end{itemize}
By the statistical extractor-based binding of $\Pi_{\mathsf{com}}$ and a standard hybrid argument, we can see that 
$P^*$ also makes $\tilde{V}$ accept with $1/2+\mathsf{non}\text{-}\mathsf{negl}(\secp)$. However, we can see that the probability must be at most $1/2$ because 
if $P^*$ can pass the verification for $c=0$, then $\tilde{G}$ is isomorphic to $G$, in which case $\tilde{G}$ does not have a Hamiltonian cycle and thus $P\*$ cannot pass the verification for $c=1$. 
Thus, the protocol satisfies statistical $1/2$-soundness. 

\paragraph{\bf Zero-knowledge.} 
Let $V^*=\{V^*_\secp,\rho_\secp\}_{\secp\in \mathbb{N}}$ be a non-uniform QPT malicious verifier. 
Then we construct a simulator $\mathsf{Sim}=\{\mathsf{Sim}_\secp,\sigma_\secp\}_{\secp\in \mathbb{N}}$ as follows;
\begin{description}
\item[$\mathsf{Sim}_\secp(\sigma_\secp,G)$:]
The quantum auxiliary input $\sigma_\secp$ is defined as $\rho_\secp^{\otimes \secp}\otimes (\ket{\psi_\secp}\bra{\psi_\secp})^{\otimes \secp}$. 

Repeat the following $\secp$ times: 
\begin{enumerate}
    \item Randomly pick $\tilde{c}\gets \bit$. 
    \item If $\tilde{c}=0$, do the following:
    \begin{enumerate}
    \item Randomly pick a permutation $\pi$ on $W$. 
    \item Commit to the adjacency matrix of $\pi(G):=\{(\pi(i),\pi(j)):(i,j)\in G\}$ using $\Pi_{\mathsf{com}}$ in a bit-wise manner. 
    \item Invoke the second round message generation by $V^*_\secp$ by giving the quantum auxiliary input $\rho_\secp$ and the commitments generated in the previous step as the first round message.
    Let $c$ be the second round message sent from $V^*_\secp$.
    \item If $c=1$, skip the rest of steps in the loop. 
    \item If $c=0$, send the reveal registers corresponding to the commitments sent in the first round and the permutation $\pi$ to $V^*_\secp$. 
    \item Run the post-processing of $V^*_\secp$, output what it outputs, and immediately halt.
    \end{enumerate}
    \item If $\tilde{c}=1$, do the following:
     \begin{enumerate}
    \item Generate a random matrix $\tilde{G}$ that has a Hamiltonian cycle.  
    \item Commit to the adjacency matrix of $\tilde{G}:$ using $\Pi_{\mathsf{com}}$ in a bit-wise manner. 
    \item Invoke the second round message generation by $V^*_\secp$ by giving the quantum auxiliary input $\rho_\secp$ and the commitments generated in the previous step as the first round message.
    Let $c$ be the second round message sent from $V^*_\secp$.
    \item If $c=0$, skip the rest of steps in the loop. 
    \item If $c=1$, send the reveal registers corresponding to the commitments corresponding to the Hamiltonian cycle of $\tilde{G}$ 
    sent in the first round to $V^*_\secp$. 
    \item Run the post-processing of $V^*_\secp$, output what it outputs, and immediately halt.
    \end{enumerate}
\end{enumerate}
If it does not halts within $\secp$ times repetitions, abort by outputting a failure symbol $\fail$. 
\end{description}

By the computational hiding of $\mathsf{com}$, $\tilde{c}$ looks uniformly random from the view of the malicous verifier $V^*_\secp$, and thus the probability that it halts in each loop is $1/2\pm \negl(\secp)$. Thus, the probability that it does not halt within $\secp$ times repetitions is $\negl(\secp)$.
Moreover, by the computational hiding of $\mathsf{com}$, it is immediate to see that the output of $\mathsf{Sim}$ conditioned on that it does not fail is computationally indistinguishable from the output of  $V^*_\secp$ in the real experiment against non-uniform QPT distinguishers. 
This implies that the protocol satisfies non-uniform computational quantum zero-knowledge. 
\end{proof}

\fi
\section{Commitments in the Common Reference Quantum State Model}\label{sec:com}
In this section, we introduce quantum commitments in the common reference quanutm state (CRQS) model and construct them unconditionally. 
Unlike the quantum auxiliary-input setting, our construction in the CRQS model satisfies both statistical hiding and statistical binding as long as the number of copies of the CRQS given to 
the malicious receiver is bounded. 

\subsection{Definition}
\begin{definition}[Quantum commitments in the CRQS model]
\label{def:CRQS}
A (non-interactive) quantum commitment scheme in the CRQS model is given by a tuple of the setup algorithm $\setup$, committer $C$, and receiver $R$, all of which are uniform QPT algorithms. 
The scheme is divided into three phases, the setup phase, commit phase, and reveal phase as follows:   
\begin{itemize}
    \item 
    {\bf Setup phase:}
    $\setup$ takes $1^\secp$ as input, uniformly samples a classical key $k\la \mathcal{K}_{\secp}$, generates two copies of the same pure state $\ket{\psi_k}$ and sends one copy each to $C$ and $R$.  
    \item
{\bf Commit phase:} 
$C$ takes $\ket{\psi_k}$ given by the setup algorithm and a bit $b\in \bit$ to commit as input, generates a quantum state on registers $\qreg{C}$ and $\qreg{R}$, and 
sends the register $\qreg{C}$ to $R$. 
\item
{\bf Reveal phase:}
$C$ sends $b$ and the register $\qreg{R}$ to $R$.
$R$ takes $\ket{\psi_k}$ given by the setup algorithm and $(b,\qreg{C},\qreg{R})$ given by $C$ as input, and outputs $b$ if it accepts and otherwise outputs $\bot$. 
\end{itemize}
As correctness, we require that $R$ accepts with probability $1$ if the protocol is run honestly.
\end{definition}
Below, we define security of commitments in the CRQS model. 
For the definition of hiding, even though the honest receiver takes only one copy of the CRQS, we consider hiding against adversaries that take many copies of the CRQS. 
This is to capture the scenario where an authority distributes many copies of the CRQS and the adversary collects some of them.  
\begin{definition}[$t$-copy statistical hiding]\label{def:hiding}
A quantum commitment scheme $(\setup,C,R)$ in the CQRS model satisfies $t$-copy statistical hiding if for any non-uniform unbounded-time algorithm $\A$, 
\begin{align*}
\left|
    \Pr[1\gets\A(1^\secp,\ket{\psi_k}^{\otimes t},\Tr_{\qreg{R}}(\sigma_{\qreg{C},\qreg{R}})):k\la \mathcal{K}_\secp,\sigma_{\qreg{C},\qreg{R}} \la C_{\mathsf{com}}(\ket{\psi_k},0)]
    -\right.\\
    \left.
    \Pr[1\gets\A(1^\secp,\ket{\psi_k}^{\otimes t},\Tr_{\qreg{R}}(\sigma_{\qreg{C},\qreg{R}})):k\la \mathcal{K}_\secp,\sigma_{\qreg{C},\qreg{R}} \la C_{\mathsf{com}}(\ket{\psi_k},1)]
    \right|\le \negl(\secp)
\end{align*}
where $C_{\mathsf{com}}$ is the commit phase of $C$. 
\end{definition}

Similarly to the quantum auxiliary-input setting, we define two notions of binding, sum-binding and extractor-based binding. 
We stress that we require those to hold against adversaries who know the classical key $k$ instead of having copies of $\ket{\psi_k}$ unlike the hiding property. This only makes the security stronger. 

\begin{definition}[Statistical sum-binding]\label{def:binding}
A quantum commitment scheme $(\setup,C,R)$ in the CQRS model satisfies statistical sum-binding if the following holds. For any pair of non-uniform unbounded-time malicious committers $C^*_0$ and $C^*_1$ that take the classical key $k$, which is sampled by the setup algorithm, as input and work in the same way in the commit phase, if we let $p_b$ to be the probability that $R$ accepts the revealed bit $b$ in the interaction with $C^*_b$ for $b\in \bit$, then we have 
\begin{align*}
p_0+p_1\le 1+\negl(\secp).
\end{align*}
\end{definition}


\begin{definition}[Extractor-based binding]\label{def:CRQS_extract}
A quantum commitment scheme $(\setup,C,R)$ in the CQRS model satisfies statistical (computational) extractor-based binding if there is a (uniform) unbounded-time algorithm $\mathcal{E}$ (called the extractor) such that for any non-uniform unbounded-time (resp. QPT) malicious committer $C^*$, 
$\mathsf{Real}_\secp^{C^*}$ and $\mathsf{Ideal}_\secp^{C^*,\mathcal{E}}$ are indistinguishable against non-uniform unbounded-time (resp. QPT) distinguishers 
where the experiments $\mathsf{Real}_\secp^{C^*}$ and   $\mathsf{Ideal}_\secp^{C^*,\mathcal{E}}$ are defined as follows.
\begin{itemize}
\item $\mathsf{Real}_\secp^{C^*}$: 
The experiment chooses $k\gets \mathcal{K}_{\secp}$ and sends $k$ and $\ket{\psi_k}$ to $C^*$ and $R$, respectively. 
The malicious committer 
$C^*$ interacts with the honest receiver $R$ in the commit and reveal phases. 
Let $b\in \{0,1,\bot\}$ be the output of $R$ and $\tau_{C^*}$ be the final state of $C^*$. The experiment outputs a tuple $(\tau_{C^*},b)$. 
\item $\mathsf{Ideal}_\secp^{C^*,\mathcal{E}}$: 
The experiment chooses $k\gets \mathcal{K}_{\secp}$ and sends $k$ and $\ket{\psi_k}$ to $C^*$ and $R$, respectively. 
The experiment also sends $k$ to the extractor $\mathcal{E}$  
The malicious committer $C^*$ runs its commit phase to generate a commitment $\sigma_{\qreg{C}}$ where $C^*$ may keep a state that is entangled with $\sigma_{\qreg{C}}$. 
$\mathcal{E}$ takes the register $\qreg{C}$, outputs an extracted bit $b^*\in \bit$, and sends a post-execution state on $\qreg{C}$ (that may be different from the original one) to $R$ as a commitment.  
Then $C^*$ and $R$ run the reveal phase.
Let $b$ be the output of $R$ and $\tau_{C^*}$ be the final state of $C^*$. 
If $b\notin \{\bot,b^*\}$, then the experiment outputs a special symbol $\fail$ and otherwise outputs a tuple $(\tau_{C^*},b)$. 
\end{itemize}
\end{definition}
\begin{remark}
Differently from the definition for quantum auxiliary-input commitments, we require the extractor to be \emph{uniform} unbounded-time algorithm that takes the classical key $k$ as input.
This is useful for constructing uniform simulators in the applications.   
\end{remark}

\begin{lemma}
Statistical (resp. computational) extractor-based binding implies statistical (resp. computational) sum-binding. 
\end{lemma}
Since the above lemma can be proven similarly to that in the quantum auxiliary-input setting (\Cref{lem:QAI_ext_to_sum}), we omit the proof.

\subsection{Construction}\label{sec:construction_CRQS}
We construct a commitment scheme in the CRQS model as follows.
Let 
$m:=\secp$\footnote{In fact, it suffices to set $m=\omega(\log \secp)$.} and
$\mathcal{H}_\secp\coloneqq\{H_k:\bit^\secp \ra \bit^{2\secp}\}_{k\in \mathcal{K}_\secp}$ be a family of $2m(t+1)$-wise independent functions.  
Then our construction is described below. 

\begin{itemize}
    \item 
    {\bf Setup phase:}
    $\setup$ takes $1^\secp$ as input and does the following. 
    It samples $k\la \mathcal{K}_{\secp}$ and generates $2m$ copies of the following states:\footnote{The state $\ket{\psi_{k,1}}$ does not depend on $k$, but we use this notation for convenience of the presentation.} 
    \begin{align*}
    \ket{\psi_{k,0}}\coloneqq\frac{1}{2^{\secp/2}}\sum_{x\in \bit^\secp}\ket{H_{k}(x)}\ket{x\concat 0^\secp},~~~\ket{\psi_{k,1}}\coloneqq\frac{1}{2^{\secp}}\sum_{y\in \bit^{2\secp}}\ket{y}\ket{y}.
    \end{align*}  
    Then it defines $\ket{\psi_k}\coloneqq\ket{\psi_{k,0}}^{\otimes m}\otimes \ket{\psi_{k,1}}^{\otimes m}$.  
    A single copy of $\ket{\psi_k}$ is sent to $C$, and a single copy of $\ket{\psi_k}$ is sent to $R$. 
    \item
{\bf Commit phase:} 
$C$ takes quantum auxiliary input $\ket{\psi_k}=\ket{\psi_{k,0}}^{\otimes m}\otimes \ket{\psi_{k,1}}^{\otimes m}$ and a bit $b\in \bit$ to commit as input. 
For each $i\in [m]$, let $\qreg{X_{i}}$ and $\qreg{Y_{i}}$ be the first and second registers of the $i$-th copy of $\ket{\psi_{k,b}}$, respectively. 
Set $\qreg{C}:=(\qreg{X_{1}},...,\qreg{X_{m}})$ and  $\qreg{R}:=(\qreg{Y_{1}},...,\qreg{Y_{m}})$
and sends the register $\qreg{C}$ to $R$.  
\item
{\bf Reveal phase:}
$C$ sends $b$ and the register $\qreg{R}:=(\qreg{Y_{1}},...,\qreg{Y_{m}})$ to $R$.
For each $i\in [m]$, let $\qreg{X'_{i}}$ and $\qreg{Y'_{i}}$ be the first and second registers of the $i$-th copy of $\ket{\psi_{k,b}}$, respectively. 
    For each $i\in[m]$, 
     $R$ runs the SWAP test between registers 
     $(\qreg{X_{i}},\qreg{Y_{i}})$ and 
     $(\qreg{X'_{i}},\qreg{Y'_{i}})$. 
      If all of the tests accept, $R$ accepts by outputting $b$ and otherwise rejects by outputting $\bot$.
\end{itemize}
If the protocol is run honestly, the registers $(\qreg{X_{i}},\qreg{Y_{i}})$ and 
     $(\qreg{X'_{i}},\qreg{Y'_{i}})$ take exactly the same state $\ket{\psi_{k,b}}$. Thus, $R$ accepts with probability $1$ and thus correctness holds.

\begin{theorem}\label{thm:hid_bin}
The above protocol satisfies $t$-copy statistical hiding and statistical extractor-based binding.  
\end{theorem}



\begin{proof}[Proof of \cref{thm:hid_bin}]~
\paragraph{\bf $t$-copy statistical hiding.}
Note that we can generate $\ket{\psi_{k,0}}$ by a single quantum oracle access to $H_k$ and $\ket{\psi_{k,1}}$ does not depend on $H_k$.
Thus, we can generate $\ket{\psi_k}=\ket{\psi_{k,0}}^{\otimes m}\otimes \ket{\psi_{k,1}}^{\otimes m}$ by $m$ quantum oracle access to $H_k$.
In the $t$-copy hiding experiment, $(t+1)$ copies of $\ket{\psi_k}$ (including the one used by the honest committer) are used. 
Since $(t+1)$ copies of $\ket{\psi_k}$ can be generated by $m(t+1)$ quantum oracle access to $H_k$, and $H_k$ is chosen from a family of $2m(t+1)$-wise independent functions, \Cref{lem:simulation_QRO} implies that the adversary's advantage in the $t$-copy hiding experiment does not change even if we replace $H_{k}$ with a uniformly random function $H$. 
After this replacement, the $t$-copy statistical hiding is shown by a direct reduction to \Cref{cor:QROM_PRG} 
via a hybrid argument 
where $t$ copies of the quantum auxiliary input given to the adversary are regarded as $\sigma_H$ in \Cref{cor:QROM_PRG}.

\paragraph{\bf Statistical extractor-based binding.}
Fix $k$, and for $y\in \bit^{2\secp}$, let $N_y$ be the number of $x\in \bit^\secp$ such that 
$H_k(x)=y$. Then we have 
\begin{align*}
F(\Tr_{\qreg{Y}}(\ket{\psi_{k,0}}\bra{\psi_{k,0}}_{\qreg{X},\qreg{Y}}),\Tr_{\qreg{Y}}(\ket{\psi_{k,1}}\bra{\psi_{k,1}}_{\qreg{X},\qreg{Y}}))
&= \left(\sum_{y\in \bit^{2\secp}}\sqrt{\frac{N_y}{2^{3\secp}}} \right)^2\\
&= \left(\sum_{y\in \mathsf{Im}(H_k)}\sqrt{\frac{N_y}{2^{3\secp}}} \right)^2\\
&\le |\mathsf{Im}(H_k)|\cdot \sum_{y\in \mathsf{Im}(H_k)} \frac{N_y}{2^{3\secp}}\\
&\le \frac{1}{2^{\secp}}
\end{align*}
where 
$\mathsf{Im}(H_k)$ donotes the image of $H_k$, 
the first inequality follows from the Cauchy–Schwarz inequality, and
the second inequality follows from $|\mathsf{Im}(H_k)|\le 2^\secp$ and $\sum_{y\in \mathsf{Im}(H_k)}N_y=2^\secp$.
We observe that $\Tr_{\qreg{Y}}(\ket{\psi_{k,0}}\bra{\psi_{k,0}})$ and $\Tr_{\qreg{Y}}(\ket{\psi_{k,1}}\bra{\psi_{k,1}})$ play the roles of $\xi_{\secp,0}$ and $\xi_{\secp,1}$ in the construction in \Cref{sec:construction_QAI} and the only assumption needed for proving statistical extractor-based binding there  was $F(\xi_{\secp,0},\xi_{\secp,1})=\negl(\secp)$. 
Thus, the rest of the proof is almost identical to that in the quantum auxiliary-input setting (\cref{thm:QAI_hid_bin})
noting that we can implement the extractor $\mathcal{E}$ by a \emph{uniform} unbounded-time algorithm given $k$. 
Thus, we omit the details.  \takashi{Is this okay? I think this is really the same.}\mor{I think so.}
\end{proof}

\ifnum\submission=1
\paragraph{\bf Impossibility of Unbounded-Copy Security.}
We have constructed a quantum commitment scheme in the CRQS model that satisfies $t$-copy statistical hiding and statistical extractor-based binding for any bounded polynomial $t$. One may wonder if this can be achieved for unbounded polynomials $t$. Unfortunately, it turns out this is impossible even if we relax the security of binding to the sum-binding against adversaries that receive $t$ copies of the CRQS (instead of receiving the classical key $k$ as in \cref{def:binding}).
See 
\ifnum\inclappndx=0
the full version of this paper~\cite{cryptoeprint:2023/1844}
\else
\Cref{sec:unbounded_copy_impossibility} 
\fi
for a proof of this impossibility.
\else
\ifnum\submission=1
\section{Impossibility of Unbounded-Copy Security in the CRQS Model}
\else
\subsection{Impossibility of Unbounded-Copy Security}
\fi
\label{sec:unbounded_copy_impossibility}
In \Cref{sec:construction_CRQS}, we construct a quantum commitment scheme in the CRQS model that satisfies $t$-copy statistical hiding and statistical extractor-based binding for any bounded polynomial $t$. A natural question is if we can achieve it for unbounded polynomials $t$. Unfortunately, it turns out this is impossible even if we relax the security of binding to the sum-binding against adversaries that receive $t$ copies of the CRQS (instead of receiving the classical key $k$ as in \cref{def:binding}). 
\begin{theorem}
There is no quantum commitment scheme in the CRQS model that satisfies statistical hiding against adversaries that get
unbounded poly-copies of CRQS, and statistical sum-binding against adversaries that get unbounded poly-copies of CRQS.
\end{theorem}
\begin{proof}
Assume that after receiving $\ket{\psi_k}$ from the setup,
the honest committer applies a unitary $U$ on $\ket{b}\ket{\psi_k}\ket{0...0}$ to generate 
$\ket{\Psi(\psi_k,b)}_{\qreg{C},\qreg{R}}\coloneqq U(\ket{b}\ket{\psi_k}\ket{0...0})$ over two registers $\qreg{R}$ and $\qreg{C}$, where
$b\in\bit$ is the bit to be committed.
Define $\sigma^{b,k}_{\qreg{C}}\coloneqq \Tr_{\qreg{R}}(\ket{\Psi(\psi_k,b)}_{\qreg{C},\qreg{R}})$.
We can show the following lemma:
\begin{lemma}
\label{lem:sigma_close}
If $|\langle \psi_k|\psi_{k'}\rangle|^2\ge1-\epsilon$, then
$\TD(\sigma^{b,k},\sigma^{b,k'})\le\sqrt{\epsilon}$.
\end{lemma}
\begin{proof}[Proof of \cref{lem:sigma_close}]
\begin{align}
\TD(\sigma^{b,k},\sigma^{b,k'})
&\le\TD(\ket{\Psi(\psi_k,b)},\ket{\Psi(\psi_{k'},b)})    \\
&=\TD(U\ket{b}\ket{\psi_k}\ket{0...0},U\ket{b}\ket{\psi_{k'}}\ket{0...0})    \\
&=\TD(\ket{\psi_k},\ket{\psi_{k'}})    \\
&\le\sqrt{\epsilon}.
\end{align}    
\end{proof}

We also use the following lemma:
\begin{lemma}
\label{lem:G}
There is a polynomial $p$ such that
\begin{align}
\frac{1}{|\mathcal{K}_\secp|}\sum_{k\in \mathcal{K}_\secp}\TD(\sigma^{0,k},\sigma^{1,k})\ge\frac{1}{p(\secp)}
\end{align}
for infinitely-many $\secp\in\mathbb{N}$.
\end{lemma}

\begin{proof}[Proof of \cref{lem:G}]
Assume that
$\frac{1}{|\mathcal{K}_\secp|}\sum_{k\in \mathcal{K}_\secp}\TD(\sigma^{0,k},\sigma^{1,k})=\negl(\secp)$.
This means that except for a negligible fraction of $k$, 
$\TD(\sigma^{0,k},\sigma^{1,k})=\negl(\secp)$.
Then we can construct the following unbounded malicious committer that breaks statistical sum-binding given unbounded-poly copies of $\ket{\psi_k}$:
\begin{enumerate}
        \item
    Do the shadow tomography on $\ket{\psi_k}^{\otimes t}$ to compute $\eta_{k'}$ for all $k'\in\mathcal{K}_\secp$
    such that $|\eta_{k'}-|\langle \psi_{k'}|\psi_{k}\rangle|^2|\le\frac{1}{r}$, where $r$ is a polynomial specified later.
    \item 
    Pick up the lexicographically first $k^*$ such that $\eta_{k^*}\ge1-\frac{1}{r}$.
    \item 
    Generate $\ket{\Psi(\psi_{k^*},0)}_{\qreg{R},\qreg{C}}$, and sends the register $\qreg{C}$ to the receiver.
    It is the end of the commit phase.
    \item 
    In the reveal phase,
    to open $b=0$, send $\qreg{R}$ to the receiver.
    To open $b=1$, apply $V^{k^*}$ on the $\qreg{R}$ register and send $\qreg{R}$ to the receiver.
\end{enumerate}
Here, $V^{k^*}$ is the unitary such that
$F(\sigma^{0,k^*}_{\qreg{C}},\sigma^{1,k^*}_{\qreg{C}})=|\langle\Psi(\psi_{k^*},1)|(V^{k^*}_{\qreg{R}}\otimes I_{\qreg{C}})|\Psi(\psi_{k^*},0)\rangle_{\qreg{R},\qreg{C}}|^2$.
Let $\Pi_{b,k}$ be the POVM element that corresponds to the acceptance of $b$ by the receiver.
The probability that the receiver opens $b=0$ is
\begin{align}
p_0&\coloneqq\frac{1}{|\mathcal{K}_\secp|}\sum_{k\in\mathcal{K}_\secp}    
\langle \Psi(\psi_{k^*},0)|\Pi_{0,k}|\Psi(\psi_{k^*},0)\rangle\\
&\ge\frac{1}{|\mathcal{K}_\secp|}\sum_{k\in\mathcal{K}_\secp}    
\langle \Psi(\psi_{k},0)|\Pi_{0,k}|\Psi(\psi_{k},0)\rangle-\sqrt{\frac{2}{r}}\\
&\ge
1-\sqrt{\frac{2}{r}}.
\end{align}
The probability that the receiver opens $b=1$ is
\begin{align}
p_1&\coloneqq\frac{1}{|\mathcal{K}_\secp|}\sum_k    
\langle \Psi(\psi_{k^*},0)|(V^{k^*,\dagger}_{\qreg{R}}\otimes I_{\qreg{C}})\Pi_{1,k}
(V^{k^*}_{\qreg{R}}\otimes I_{\qreg{C}})|\Psi(\psi_{k^*},0)\rangle\\
&\ge\frac{1}{|\mathcal{K}_\secp|}\sum_k    
\langle \Psi(\psi_{k^*},1)|\Pi_{1,k}
|\Psi(\psi_{k^*},1)\rangle\\
&-\frac{1}{|\mathcal{K}_\secp|}\sum_k\TD\left((V^{k^*}_{\qreg{R}}\otimes I_{\qreg{C}})\ket{\Psi(\psi_{k^*},0)},\ket{\Psi(\psi_{k^*},1)}\right)    \\
&\ge\frac{1}{|\mathcal{K}_\secp|}\sum_k    
\langle \Psi(\psi_{k},1)|\Pi_{1,k}
|\Psi(\psi_{k},1)\rangle-\sqrt{\frac{2}{r}}\\
&-\frac{1}{|\mathcal{K}_\secp|}\sum_k
\sqrt{1-|\langle\Psi(\psi_{k^*},1)|(V^{k^*}_{\qreg{R}}\otimes I_{\qreg{C}})\ket{\Psi(\psi_{k^*},0)}|^2}   \\
&\ge1-\sqrt{\frac{2}{r}}
-\frac{1}{|\mathcal{K}_\secp|}\sum_k
\sqrt{1-F(\sigma^{0,k^*},\sigma^{1,k^*})}   \\
&\ge1-\sqrt{\frac{2}{r}}
-\frac{1}{|\mathcal{K}_\secp|}\sum_k
\sqrt{1-[1-\TD(\sigma^{0,k^*},\sigma^{1,k^*})]^2} \\
&\ge1-\sqrt{\frac{2}{r}}
-\frac{1}{|\mathcal{K}_\secp|}\sum_k
\sqrt{1-\left[1-\TD(\sigma^{0,k},\sigma^{1,k})-2\sqrt{\frac{2}{r}}\right]^2} \\
&\ge1-\sqrt{\frac{2}{r}}
-(1-\negl(\secp))
\sqrt{1-\left[1-\negl(\secp)-2\sqrt{\frac{2}{r}}\right]^2} \\
&\ge1-\sqrt{\frac{2}{r}}
-2\left(\frac{2}{r}\right)^{\frac{1}{4}}-\negl(\secp).
\end{align}
Hence if we take $r\gg1$, we have
\begin{align}
p_0+p_1&\ge    
1-\sqrt{\frac{2}{r}}
+
1-\sqrt{\frac{2}{r}}
-2\left(\frac{2}{r}\right)^{\frac{1}{4}}\\
&\ge1+\frac{1}{2},
\end{align}
which breaks the sum binding.
\end{proof}

For each $k\in\mathcal{K}_\secp$, define a POVM measurement $\{\Lambda_k,I-\Lambda_k\}$ such that
$\TD(\sigma^{0,k},\sigma^{1,k})=\Tr(\Lambda_k\sigma^{0,k})-\Tr(\Lambda_k\sigma^{1,k})$. 
Let us consider the following unbounded adversary $\cA$ that breaks the statistical hiding:
\begin{enumerate}
    \item
    Do the shadow tomography on $\ket{\psi_k}^{\otimes t}$ to compute $\eta_{k'}$ for all $k'\in\mathcal{K}_\secp$
    such that $|\eta_{k'}-|\langle \psi_{k'}|\psi_{k}\rangle|^2|\le\frac{1}{u}$, where $u$ is a polynomial specified later.
    \item 
    Pick up the lexicographically first $k^*$ such that $\eta_{k^*}\ge1-\frac{1}{u}$.
    \item 
    Measure $\sigma^{b,k}_{\qreg{C}}$, which is committed by the committer, with $\{\Lambda_{k^*},I-\Lambda_{k^*}\}$.
    \item 
    If the result $\Lambda_{k^*}$ is obtained, output 0.
    Otherwise, output 1.
\end{enumerate}
Then,
\begin{align}
&\Pr[0\gets\cA|b=0]
-\Pr[0\gets\cA|b=1]\\
&=\frac{1}{|\mathcal{K}_\secp|}\sum_{k\in\mathcal{K}_\secp}\Tr(\Lambda_{k^*}\sigma^{0,k})-
\frac{1}{|\mathcal{K}_\secp|}\sum_{k\in\mathcal{K}_\secp}\Tr(\Lambda_{k^*}\sigma^{1,k})\\
&\ge\frac{1}{|\mathcal{K}_\secp|}\sum_{k\in\mathcal{K}_\secp}\Tr(\Lambda_{k^*}\sigma^{0,k^*})-
\frac{1}{|\mathcal{K}_\secp|}\sum_{k\in\mathcal{K}_\secp}\Tr(\Lambda_{k^*}\sigma^{1,k^*})-2\sqrt{\frac{2}{u(\secp)}}\\
&=\frac{1}{|\mathcal{K}_\secp|}\sum_{k\in\mathcal{K}_\secp}\TD(\sigma^{0,k^*},\sigma^{1,k^*})-2\sqrt{\frac{2}{u(\secp)}}\\
&\ge\frac{1}{|\mathcal{K}_\secp|}\sum_{k\in\mathcal{K}_\secp}\TD(\sigma^{0,k},\sigma^{1,k})-4\sqrt{\frac{2}{u(\secp)}}\\
&\ge\frac{1}{\poly(\secp)},
\end{align}
where in the last inequality, we have used \cref{lem:G} and taken $u\gg p$.
Hence the statistical hiding is broken.
\end{proof}

\fi

\paragraph{\bf Circumventing the impossibility using stateful setup.}
Interestingly, we can circumvent the above impossibility if we allow the setup algorithm to be stateful.%
\ifnum\anonymous=0
\footnote{We thank Fermi Ma for suggesting this.}
\fi
\takashi{more explanation on the model may be useful. I don't want to write the formal description since that may look too complicated.}
To see this, we first observe that if $H_k$ is replaced with a uniformly random function from $\bit^\secp$ to $\bit^{2\secp}$ 
in the construction in \cref{sec:construction_CRQS}, then it satisfies $t$-copy statistical hiding for all polynomials (or even subexponential) $t$. However, such a modified protocol has inefficient setup algorithm since a random function cannot be computed efficiently.
A common solution in such a situation is to use pseudorandom functions, but then the hiding becomes a computational one.
This is not useful for our purpose since if we can use pseudorandom functions, we could directly construct computationally hiding and statistically binding quantum commitments in the plain model.  
Another common method to efficiently simulating a random function is to rely on lazy-sampling, i.e., instead of sampling the whole function at the beginning, we assign the function values only on queried inputs.  
Recently, Zhandry~\cite{C:Zhandry20} proposed a technique called the compressed oracle that enables us to perfectly and efficiently simulate a quantumly-accessible random oracle.     
Thus, if the setup algorithm uses the compressed oracle technique to simulate the random function, we can achieve $t$-copy statistical hiding for all polynomials $t$ and statistical extractor-based binding simultaneously,\footnote{
In the stateful setup setting, we have to slightly modify the definition of statistical extractor-based binding since the classical key $k$ no longer appears. We allow the malicious committer and extractor to receive arbitrarily many (possibly exponential number of) copies of the CRQS. The proof still works in this setting essentially in the same way.  \takashi{I added this footnote.}
} at the cost of making the setup algorithm stateful so that it can keep the quantum ``database'' needed for the simulation of the random oracle. \takashi{We may polish the explanation here later.} 

\subsection{Applications of Commitments in the CRQS model}
In the plain model (where there is no CRQS), \cite{C:BCKM21b} constructed oblivious transfers (OTs) from any post-quantum classical commitments.  
\cite{C:AnaQiaYue22} observed that the construction works based on any quantum commitments that satisfy a binding property they introduced.\footnote{Later, Yan~\cite{AC:Yan22} showed that any canonical quantum bit commitment scheme satisfies the binding property defined in \cite{C:AnaQiaYue22}.} 
Since our definition of the extractor-based binding closely follows the definition of binding in \cite{C:AnaQiaYue22},  
we observe that we can similarly plug our commitments in the CRQS model into the compiler of \cite{C:BCKM21b} to obtain a statistically secure oblivious transfer in the CRQS model. 
Thus, we obtain the following corollary. 
\begin{corollary}
For any polynomial $t$, 
there exist $t$-copy statistically maliciously simulation-secure OTs in the CRQS model. 
\end{corollary}
For clarity, we describe our definition of OTs in the CRQS model 
in 
\ifnum\inclappndx=0
the full version of the paper~\cite{cryptoeprint:2023/1844}%
\else
\Cref{def:OT_CRQS}%
\fi
. 
\begin{remark}
One may wonder why the compiler of \cite{C:BCKM21b} works in the CRQS model but does not work in the quantum auxiliary-input setting as considered in \Cref{sec:QAI}. Roughly, this is because the CRQS can be efficiently computed given the classical key $k$, but there is no way to efficiently compute the quantum auxiliary input. The lack of efficient generation of the quantum auxiliary input prevents us from applying Watrous' rewinding lemma~\cite{SIAM:Watrous09}, which is a crucial tool in the compiler in \cite{C:BCKM21b}.    
\end{remark}

Moreover it is known that OTs imply MPCs for classical functionalities~\cite{C:IshPraSah08} or even for quantum functionalities~\cite{EC:DGJMS20} in a black-box manner. Thus, we believe that similar constructions work in the CRQS model, which would lead to  statistically maliciously simulation-secure MPCs in the CRQS model. However, for formally stating it, we have to carefully reexamine these constructions to make sure that they work in the CRQS model as well. This is out of the scope of this work, and we leave it to future work. \takashi{I'm not super familiar with MPC and cannot state these results confidently.}


\ifnum\anonymous=1
\else
\ifnum\llncs=1
\begin{credits}
\subsubsection{\ackname}
\else
\paragraph{Acknowledgments.}
\fi
We thank Taiga Hiroka for helpful discussions, 
Qipeng Liu for answering questions regarding \cite{EC:Liu23}, and
Fermi Ma for many insightful comments including suggestion of an alternative proof of \Cref{thm:pq-sparse_pseudorandom} and the idea of using the compressed oracle to achieve unbounded-copy security. 
We thank Luowen Qian for sharing an early draft of \cite{Qian23} 
and providing many helpful comments on our earlier draft, especially pointing out flaws in the proofs of \Cref{thm:pq-sparse_pseudorandom,thm:QAI_hid_bin}. 
This work was done in part while T. Morimae and B. Nehoran were visiting the Simons Institute for the Theory of Computing.
This work was done in part while B. Nehoran was visiting the Yukawa Institute for Theoretical Physics.
TM is supported by
JST CREST JPMJCR23I3,
JST Moonshot R\verb|&|D JPMJMS2061-5-1-1, 
JST FOREST, 
MEXT QLEAP, 
the Grant-in Aid for Transformative Research Areas (A) 21H05183,
and 
the Grant-in-Aid for Scientific Research (A) No.22H00522.
\ifnum\llncs=1
\subsubsection{\discintname}
The authors have no competing interests to declare that are
relevant to the content of this article.
\end{credits}
\else
\fi
\fi

\ifnum\submission=0
\bibliographystyle{alpha} 
\else
\bibliographystyle{splncs04}
\fi
\bibliography{abbrev3,crypto,reference}

\ifnum\inclappndx=1
\appendix 
\ifnum\submission=1
    
\fi
\section{No-go in Classical Setup Models}
\label{sec:nogo_correlatedRS}

\begin{definition}[Commitments in the correlated randomness model]
A (non-interactive) commitment scheme in the correlated randomness model is given by a tuple of the setup algorithm 
$\setup$, committer $C$, and receiver $R$, all of which are uniform QPT algorithms, and a family $\{D_\secp\}_{\secp\in\mathbb{N}}$ of 
distributions over classical bit strings.
The scheme is divided into three phases, the setup phase, commit phase, and reveal phase as follows:  
\begin{itemize}
    \item 
    {\bf Setup phase:}
    $\setup$ takes $1^\secp$ as input and samples $(x,y)\gets D_\secp$. 
    It sends $x$ to $C$ and $y$ to $R$.  
    \item
{\bf Commit phase:} 
$C$ takes $x$ given by the setup algorithm and a bit $b\in \bit$ to commit as input, generates a quantum state $\ket{\Psi_{x,b}}_{\qreg{C},\qreg{R}}$ on registers $\qreg{C}$ and $\qreg{R}$, and 
sends the register $\qreg{C}$ to $R$. 
\item
{\bf Reveal phase:}
$C$ sends $b$ and the register $\qreg{R}$ to $R$.
$R$ does a verification POVM measurement $\{\Lambda_{y,b},I-\Lambda_{y,b}\}$ on
the state over $\qreg{C}$ and $\qreg{R}$.
If $\Lambda_{y,b}$ is obtained, $R$ outputs $b$. Otherwise, $R$ outputs $\bot$. 
\end{itemize}
\end{definition}
The correctness, statistical hiding, and statistical binding are defined in similar ways as
those of commitments in the CRQS model (\cref{def:CRQS}).

\begin{remark}
The collapsing theorem~\cite{AC:Yan22} will not hold in general in the correlated randomness model,
because, for example, the committer cannot know the value of $y$ and therefore the entire commitment phase cannot be delegated to the committer.
In the CRS model, on the other hand, we can show the collapsing theorem in a similar way as \cite{AC:Yan22}.
\end{remark}

\begin{definition}
We say that a family of distributions $\{D_\secp\}_{\secp\in\mathbb{N}}$ is $\epsilon$-correlated if
there is a (not necessarily efficient) algorithm $\cA$ such that
\begin{align}
\Pr[x\gets\cA(y):(x,y)\gets D_\secp]\ge \epsilon.
\label{correlated}
\end{align}
\end{definition}

\begin{theorem}
\label{thm:correlated}
If $\{D_\secp\}_{\secp\in\mathbb{N}}$ is $\epsilon$-correlated with $\epsilon\ge\frac{1}{2}+\frac{1}{\poly(\secp)}$,
statistically-hiding and statistically-binding (non-interactive) commitments in the correlated randomness model is impossible.
\end{theorem}

\begin{remark}
This in particular means that
statistically-hiding and statistically-binding commitments in the common reference string (CRS) model
is impossible. Here, the CRS model is a special case of the correlated randomness model where
a setup algorithm samples a bit string $x$ from a distribution, and sends the same $x$ to both the committer and the receiver.
\end{remark}

\begin{remark}
\label{rem:symmertric-shared-randomness}
The following symmetric randomness setup (which is a classical analogue of the CRQS model) is a special case of the correlated randomness model.
Therefore, statistically-secure commitments are also impossible with the following setup if it is $\epsilon$-correlated with $\epsilon\ge\frac{1}{2}+\frac{1}{\poly(\secp)}$.     
\begin{enumerate}
    \item 
    Sample $k\gets\cK_\secp$.
    \item 
    Sample $x\gets E_k$.
    Sample $y\gets E_k$.
    Here $\{E_k\}_{k\in\cK_\secp}$ is a family of certain distributions.
    \item 
    Output $(x,y)$.
\end{enumerate}
\end{remark}

\begin{remark}
For the general case, this result is essentially tight, because statistically-hiding and statistically-binding commitments are possible (even completely classically) if
the setup algorithm samples two bit strings $x_0,x_1$ and a single bit $c$ uniformly at random, and sends
$(c,x_c)$ to the committer and $(x_0,x_1)$ to the receiver.
Such a distribution is $\frac{1}{2}$-correlated.

On the other hand, this is not necessarily tight in the specific case of symmetric shared randomness, where all parties must receive samples from the same distribution (as in \Cref{rem:symmertric-shared-randomness}). While we are not aware of any construction, even if $\epsilon \le \frac{1}{2}$, we do not rule one out. We leave as an open problem to show an impossibility in this case.
\end{remark}

\begin{proof}[Proof of \cref{thm:correlated}]
For each $x$, define
$\rho_{x,b}\coloneqq\Tr_{\qreg{R}}(\ket{\Psi_{x,b}}_{\qreg{R},\qreg{C}})$.
For each $x$, let $\{\Pi_x,I-\Pi_x\}$ be a POVM measurement such that
\begin{align}
\TD(\rho_{x,0},\rho_{x,1})=\Tr[\Pi_x(\rho_{x,0}-\rho_{x,1})].    
\label{tracePOVM}
\end{align}
By assumption, there exists an algorithm $\cA$ and a polynomial $p$ such that
\begin{align}
\Pr[x\gets\cA(y):(x,y)\gets D_\secp]\ge \frac{1}{2}+\frac{1}{p(\secp)}.    
\label{correlationassumption}
\end{align}
Let us consider the following adversary for the statistical hiding.
\begin{enumerate}
    \item 
    Get $y$ and the register $\qreg{C}$ as input.
    \item 
    Run $x'\gets\cA(y)$.
    \item 
    Measure the register $\qreg{C}$ with the POVM measurement $\{\Pi_{x'},I-\Pi_{x'}\}$. 
    If $\Pi_{x'}$ is obtained, output 0.
    Otherwise, output 1.
\end{enumerate}
Let $p_{x,y}$ be the probability that $(x,y)$ is sampled from $D_\secp$.
The probability that the adversary correctly answers 0 when $b=0$ is
at least
\begin{align}
\sum_{x,y}p_{x,y}\Pr[x\gets\cA(y)]\Tr[\Pi_{x}\rho_{x,0}]    
\end{align}
and
the probability that the adversary correctly answers 1 when $b=1$ is
at least
\begin{align}
\sum_{x,y}p_{x,y}\Pr[x\gets\cA(y)]\Tr[(I-\Pi_{x})\rho_{x,1}].
\end{align}
Because of the statistical hiding, 
\begin{align}
\sum_{x,y}p_{x,y}\Pr[x\gets\cA(y)]\Tr[\Pi_{x}\rho_{x,0}]    
+\sum_{x,y}p_{x,y}\Pr[x\gets\cA(y)]\Tr[(I-\Pi_{x})\rho_{x,1}]
\le1+\negl(\secp),
\end{align}
which means
\begin{align}
\sum_{x,y}p_{x,y}\Pr[x\gets\cA(y)]\Tr[\Pi_{x}(\rho_{x,0}-\rho_{x,1})]    
\le
1-
\sum_{x,y}p_{x,y}\Pr[x\gets\cA(y)]
+\negl(\secp).
\label{aa}
\end{align}
Therefore
\begin{align}
\sum_{x,y}p_{x,y}\Pr[x\gets\cA(y)]\left(1-\sqrt{F(\rho_{x,0},\rho_{x,1})}\right)
&\le    
\sum_{x,y}p_{x,y}\Pr[x\gets\cA(y)]\TD(\rho_{x,0},\rho_{x,1})    \\
&\le1-
\sum_{x,y}p_{x,y}\Pr[x\gets\cA(y)]
+\negl(\secp),
\label{bb}
\end{align}
where we have used the relation $1-\sqrt{F(\rho,\sigma)}\le\TD(\rho,\sigma)$,
\cref{tracePOVM}, and \cref{aa}.
This means
\begin{align}
\sum_{x,y}p_{x,y}\sqrt{F(\rho_{x,0},\rho_{x,1})}&\ge    
\sum_{x,y}p_{x,y}\Pr[x\gets\cA(y)]\sqrt{F(\rho_{x,0},\rho_{x,1})}\\
&\ge2\sum_{x,y}p_{x,y}\Pr[x\gets\cA(y)]
-1
-\negl(\secp)\\
&\ge\frac{2}{p(\secp)}-\negl(\secp)\\
&\ge\frac{1}{p(\secp)}.
\end{align}
Here we have used \cref{bb} and \cref{correlationassumption}.
Define
\begin{align}
G\coloneqq\left\{(x,y):\sqrt{F(\rho_{x,0},\rho_{x,1})}\ge\frac{1}{2p}\right\}.    
\end{align}
Then the standard average argument shows that
$\sum_{(x,y)\in G}p_{x,y}\ge \frac{1}{2p}$.

Due to the correctness,
\begin{align}
\sum_{x,y}p_{x,y}\bra{\Psi_{x,1}}\Lambda_{y,1}\ket{\Psi_{x,1}}\ge1-\negl(\secp).    
\end{align}
If we define 
\begin{align}
   T\coloneqq\left\{(x,y):
\bra{\Psi_{x,1}}\Lambda_{y,1}\ket{\Psi_{x,1}}\ge \sqrt{1-\frac{1}{5p^2}}    
   \right\}, 
   \label{defT}
\end{align}
the standard average argument shows that
$\sum_{(x,y)\in T}p_{x,y}\ge1-\negl(\secp)$.
Moreover, from the union bound,
\begin{align}
\sum_{(x,y)\in G\cap T}p_{x,y}\ge \frac{1}{2p}-\negl(\secp).    
\label{union}
\end{align}

Let $U^x$ be a unitary such that
$F(\rho_{x,0},\rho_{x,1})=|\bra{\Psi_{x,1}}(U_{\qreg{R}}^x\otimes I_{\qreg{C}})\ket{\Psi_{x,0}}|^2$.
Let us consider the following adversary $\cB$ for the statistical binding.
\begin{enumerate}
\item 
Receive $x$ from the setup algorithm.
    \item 
    Commit $b=0$ honestly.
    \item 
    To open $b=0$, just send $\qreg{R}$.
    \item 
    To open $b=1$, apply $U^x$ on $\qreg{R}$ and send $\qreg{R}$.
\end{enumerate}
Due to the correctness, the probability that $R$ opens 0 is at least $1-\negl(\secp)$.
The probability that $R$ opens 1 is
\begin{align}
\sum_{x,y}p_{x,y} \bra{\Psi_{x,0}}U^{x\dagger} \Lambda_{y,1}U^x\ket{\Psi_{x,0}}   
&\ge
\sum_{(x,y)\in G\cap T}p_{x,y} \bra{\Psi_{x,0}}U^{x\dagger} \Lambda_{y,1}U^x\ket{\Psi_{x,0}}   \\
&\ge
\sum_{(x,y)\in G\cap T}p_{x,y} \left[\bra{\Psi_{x,1}} \Lambda_{y,1}\ket{\Psi_{x,1}} 
-\sqrt{1-\frac{1}{4p^2}}\right]  \\
&\ge
\sum_{(x,y)\in G\cap T}p_{x,y} \left[\sqrt{1-\frac{1}{5p^2}} 
-\sqrt{1-\frac{1}{4p^2}}\right]  \\
&\ge
\left(\frac{1}{2p}-\negl(\secp)\right) \left[\sqrt{1-\frac{1}{5p^2}} 
-\sqrt{1-\frac{1}{4p^2}}\right]\\
&\ge\frac{1}{\poly(\secp)}.
\end{align}
Here we have used \cref{defT} and \cref{union}.
Also, in the second inequality, 
we have used the fact that
\begin{align}
|\bra{\Psi_{x,0}}U^{x\dagger}\Lambda_{y,1}U^x\ket{\Psi_{x,0}} 
-\bra{\Psi_{x,1}}\Lambda_{y,1}\ket{\Psi_{x,1}}|
&\le
\TD(U^x\ket{\Psi_{x,0}}\bra{\Psi_{x,0}}U^{x\dagger},
\ket{\Psi_{x,1}}\bra{\Psi_{x,1}})\\
&\le
\sqrt{1-F(U^x\ket{\Psi_{x,0}},\ket{\Psi_{x,1}})}\\
&=
\sqrt{1-F(\rho_{x,0},\rho_{x,1})}\\
&\le
\sqrt{1-\frac{1}{4p^2}}.
\end{align}
for any $(x,y)\in G$. 
In conclusion, the sum-binding is broken by $\cB$.
\end{proof}

\begin{theorem}
Statistically-hiding and statistically-binding (non-interactive) commitments
in the correlated randomness model with $\{D_\secp\}_{\secp\in\mathbb{N}}$
is impossible if $\{D_\secp\}_{\secp\in\mathbb{N}}$
satisfies
$\|D_\secp-E_\secp\times F_\secp\|_1<\frac{5-2\sqrt{2}}{17}=0.1277...$
for some families of distributions $\{E_\secp\}_{\secp\in\mathbb{N}}$ and $\{F_\secp\}_{\secp\in\mathbb{N}}$.
\end{theorem}

\begin{proof}
Assume that there is a statistically-hiding and statistically-binding (non-interactive) commitment scheme
in the correlated randomness model with $\{D_\secp\}_{\secp\in\mathbb{N}}$ such that
$\|D_\secp-E_\secp\times F_\secp\|_1\le\epsilon$
for some families of distributions $\{E_\secp\}_{\secp\in\mathbb{N}}$ and $\{F_\secp\}_{\secp\in\mathbb{N}}$,
where $\epsilon=\frac{5-2\sqrt{2}}{17}$.

Let $p_{x,y}$ be the probability that $(x,y)$ is sampled from $D_\secp$.
Let $e_x$ be the probability that $x$ is sampled from $E_\secp$.
Let $f_y$ be the probability that $y$ is sampled from $F_\secp$.
The assumption
$\|D_\secp-E_\secp\times F_\secp\|_1\le\epsilon$
means that 
\begin{align}
\sum_{x,y}|p_{x,y}-e_xf_y|\le\epsilon.
\label{nocorrelation}
\end{align}

Let $\rho_{x,b}\coloneqq \Tr_{\qreg{R}}(\ket{\Psi_{x,b}}_{\qreg{R},\qreg{C}})$,
where $\ket{\Psi_{x,b}}_{\qreg{R},\qreg{C}}$ is the state that the honest committer generates to commit $b$.
Due to the statistical hiding,
\begin{align}
1-\sqrt{F\left(\sum_xe_x\rho_{x,0},\sum_xe_x\rho_{x,1}\right)}
&\le
\frac{1}{2}\left\|\sum_{x}e_x \rho_{x,0}-\sum_x e_x \rho_{x,1} \right\|_1  \\
&=
\frac{1}{2}\left\|(\sum_{x}e_x \rho_{x,0})\otimes(\sum_y f_y \ket{y}\bra{y})
-(\sum_x e_x \rho_{x,1})\otimes(\sum_y f_y \ket{y}\bra{y}) \right\|_1  \\
&=
\frac{1}{2}\left\|\sum_{x,y}e_xf_y \rho_{x,0}\otimes\ket{y}\bra{y}
-\sum_{x,y} e_xf_y \rho_{x,1}\otimes \ket{y}\bra{y} \right\|_1  \\
&\le
\sum_{x,y}|p_{x,y}-e_xf_y|\\
&~~~+\frac{1}{2}\left\|\sum_{x,y}p_{x,y} \rho_{x,0}\otimes\ket{y}\bra{y}
-\sum_{x,y} p_{x,y} \rho_{x,1}\otimes \ket{y}\bra{y} \right\|_1  \\
&\le \epsilon+\negl(\secp),
\end{align}
which means
\begin{align}
   F\left(\sum_x e_x\rho_{x,0} ,\sum_x e_x\rho_{x,1}\right)
   \ge (1-\epsilon)^2-\negl(\secp).
   \label{Flowerbound}
\end{align}
Define the state
    \begin{align}
    \ket{\Phi_b}_{\qreg{A},\qreg{R},\qreg{C}}\coloneqq   \sum_{x}\sqrt{e_{x}}\ket{x}_{\qreg{A}}\otimes\ket{\Psi_{x,b}}_{\qreg{R},\qreg{C}}.
    \end{align}
Due to the Uhlmann's theorem and \cref{Flowerbound}, there exists a unitary $U$ acting on registers $\qreg{A}$ and $\qreg{R}$ such that
\begin{align}
|\bra{\Phi_1}    U_{\qreg{A},\qreg{R}} \ket{\Phi_0}_{\qreg{A},\qreg{R},\qreg{C}}|^2\ge (1-\epsilon)^2-\negl(\secp).
\label{innerproduct}
\end{align}

Let us consider the following malicious committer:
\begin{enumerate}
    \item 
    Get $x$ as input (and ignore it).
    \item 
    Generate the state $\ket{\Phi_0}_{\qreg{A},\qreg{R},\qreg{C}}$,
    and send the register $\qreg{C}$ to the receiver as the commitment.
    \item 
    To open $b=0$, just send the register $\qreg{R}$ to the receiver.
    \item 
    To open $b=1$, apply $U$ on the registers $\qreg{A}$ and $\qreg{R}$, and send the register $\qreg{R}$ to the receiver.
\end{enumerate}
The probability that the receiver opens $b=0$ is
\begin{align}
   \sum_{x,y}p_{x,y}\sum_{x'}e_{x'}\bra{\Psi_{x',0}}\Lambda_{y,0}\ket{\Psi_{x',0}} 
   &\ge
   \sum_{x,y}e_xf_y\sum_{x'}e_{x'}\bra{\Psi_{x',0}}\Lambda_{y,0}\ket{\Psi_{x',0}} -\epsilon\\
   &=
   \sum_{x,y}e_xf_y\bra{\Psi_{x,0}}\Lambda_{y,0}\ket{\Psi_{x,0}} -\epsilon\\
   &\ge
   \sum_{x,y}p_{x,y}\bra{\Psi_{x,0}}\Lambda_{y,0}\ket{\Psi_{x,0}} -2\epsilon\\
   &\ge1-\negl(\secp)-2\epsilon,
\end{align}
where we have used the correctness and \cref{nocorrelation}.
The probability that the receiver opens $b=1$ is
\begin{align}
   \sum_{x,y}p_{x,y}\bra{\Phi_0}U^\dagger \Lambda_{y,1}U\ket{\Phi_{0}} 
   &\ge
   \sum_{x,y}p_{x,y}\bra{\Phi_1} \Lambda_{y,1}\ket{\Phi_1} 
   -\sqrt{2\epsilon-\epsilon^2}\\
    &\ge
   \sum_{x,y}e_xf_y\bra{\Phi_1} \Lambda_{y,1}\ket{\Phi_1} 
   -\sqrt{2\epsilon-\epsilon^2}-\epsilon\\
     &=
   \sum_{y}f_y\bra{\Phi_1} \Lambda_{y,1}\ket{\Phi_1} 
   -\sqrt{2\epsilon-\epsilon^2}-\epsilon\\
      &=
   \sum_{y}f_y
   \sum_x e_x
   \bra{\Psi_{x,1}} \Lambda_{y,1}\ket{\Psi_{x,1}} 
   -\sqrt{2\epsilon-\epsilon^2}-\epsilon\\
       &\ge
   \sum_{x,y}p_{x,y}
   \bra{\Psi_{x,1}} \Lambda_{y,1}\ket{\Psi_{x,1}} 
   -\sqrt{2\epsilon-\epsilon^2}-2\epsilon\\
        &\ge
   1-\negl(\secp)-\sqrt{2\epsilon-\epsilon^2}-2\epsilon,
\end{align}
where we have used the correctness, \cref{nocorrelation}, and \cref{innerproduct}.
Therefore if $\epsilon<\frac{5-2\sqrt{2}}{17}$, the statistical sum-binding is broken.
\end{proof}

\ifnum\submission=1
    
\fi
\section{Definition of OTs in the CRQS model}\label{def:OT_CRQS}
We define OTs in the CRQS model.
\begin{definition}[OTs in the CRQS model] 
An oblivious transfer (OT) in the CRQS model is given by a tuple of the setup algorithm $\setup$, sender $S$, and receiver $R$, all of which are uniform QPT algorithms. 
The scheme is divided into two phases, the setup phase and execution phase as follows:   
\begin{itemize}
    \item 
    {\bf Setup phase:}
    $\setup$ takes $1^\secp$ as input, uniformly samples a classical key $k\la \mathcal{K}_{\secp}$, generates two copies of the same pure state $\ket{\psi_k}$ and sends one copy each to $S$ and $R$.  
    \item
{\bf Execution phase:} 
In addition to $\ket{\psi_k}$ sent from the setup algorithm,  
$S$ takes a pair of classical messages $(m_0,m_1)$
and $R$ takes a bit $b$ as their private inputs, respectively. 
$S$ and $R$ interact through quantum channel, after which 
$S$ outputs $\top$ or $\bot$ and 
$R$ outputs a classical message $m'$ or $\bot$. 
\end{itemize}
As correctness, we require that 
$S$ outputs $\top$ and 
$R$ outputs $m_b$ with probability $1$ if the protocol is run honestly.
\end{definition}

\begin{definition}[Security of OTs in the CRQS model]
We say that an OT in the CRQS model $(\setup,S,R)$ is $t$-copy statistically maliciously simulation-secure if it satisfies the following.

\paragraph{\bf $t$-copy statistical receiver security:}
There is a uniform QPT algorithm $\mathsf{Sim}$ (called the simulator) such that for any non-uniform unbounded-time sender $S^*=\{S^*_\secp,\rho_\secp\}_{\secp\in \mathbb{N}}$ and any bit $b$,  
$$
\TD(\mathsf{Real}\text{-}\mathsf{R}_\secp^{S^*,b,t},\mathsf{Ideal}\text{-}\mathsf{R}_\secp^{S^*,\mathsf{Sim},b,t})\le \negl(\secp) 
$$
where the experiments $\mathsf{Real}\text{-}\mathsf{R}_\secp^{S^*,b,t}$ and $\mathsf{Ideal}\text{-}\mathsf{R}_\secp^{S^*,\mathsf{Sim},b,t}$
are defined as follows.
\begin{itemize}
\item $\mathsf{Real}\text{-}\mathsf{R}_\secp^{S^*,b,t}$: 
The experiment chooses $k\gets \mathcal{K}_{\secp}$ and sends $t$ copies of $\ket{\psi_k}$ to both $S^*$ and $R$.  
The malicious sender
$S^*$ interacts with the honest receiver $R$ with the private input $b$ in the execution phase. 
Let $m'$ be the output of $R$ and $\tau_{S^*}$ be the final state of $S^*$. The experiment outputs a tuple $(\tau_{S^*},m')$. 
\item $\mathsf{Ideal}\text{-}\mathsf{R}_\secp^{S^*,\mathsf{Sim},b,t}$: 
$\mathsf{Sim}$ takes $(S^*_\secp,\rho_\secp)$ as input and
outputs a state $\tilde{\tau}_{S^*}$ along with a pair of messages $(m_0,m_1)$ or $\mathsf{abort}$. 
In the case of $\mathsf{abort}$, then 
the experiment outputs a tuple  $(\tilde{\tau}_{S^*},\bot)$ 
and otherwise outputs $(\tilde{\tau}_{S^*},m_b)$.  
\end{itemize}

\paragraph{\bf $t$-copy statistical sender security:}
There is a uniform QPT algorithm $\mathsf{Sim}$ (called the simulator) such that for any non-uniform unbounded-time malicious receiver $R^*=\{R^*_\secp,\rho_\secp\}_{\secp\in \mathbb{N}}$ and any pair of messages  $(m_0,m_1)$,  
$$
\TD(\mathsf{Real}\text{-}\mathsf{S}_\secp^{R^*,(m_0,m_1),t},\mathsf{Ideal}\text{-}\mathsf{S}_\secp^{\mathsf{Sim},(m_0,m_1),t})\le \negl(\secp) 
$$
where the experiments $\mathsf{Real}\text{-}\mathsf{S}_\secp^{R^*,(m_0,m_1),t}$ and $\mathsf{Ideal}\text{-}\mathsf{S}_\secp^{R^*,\mathsf{Sim},(m_0,m_1),t}$
are defined as follows.
\begin{itemize}
\item $\mathsf{Real}\text{-}\mathsf{S}_\secp^{R^*,(m_0,m_1),t}$: 
The experiment chooses $k\gets \mathcal{K}_{\secp}$ and sends $t$ copies of $\ket{\psi_k}$ to both $S$ and $R^*$.  
The malicious sender
$R^*$ interacts with the honest sender $S$ with the private input $(m_0,m_1)$ in the execution phase. 
Let 
$d\in \{\top,\bot\}$ be the output of $S$ and  $\tau_{R^*}$ be the final state of $R^*$. The experiment outputs a tuple $(\tau_{S^*},d)$. 
\item $\mathsf{Ideal}\text{-}\mathsf{S}_\secp^{R^*,\mathsf{Sim},(m_0,m_1),t}$: 
$\mathsf{Sim}$ takes $(R^*_\secp,\rho_\secp)$ as input and
outputs a state $\tilde{\tau}_{R^*}$ and $\tilde{d}\in \{\top,\bot\}$
The experiment outputs a tuple  $(\tilde{\tau}_{R^*},\tilde{d})$. 
\end{itemize}
\end{definition}

\fi

\end{document}